\documentclass[a4paper,11pt,runningheads]{article}

\usepackage{algorithm}
\usepackage{algpseudocode}
\usepackage{graphicx} 
\usepackage{pdfpages}
\usepackage{amssymb,amsmath}
\usepackage{amsthm}
\usepackage[colorinlistoftodos]{todonotes}
\usepackage[utf8]{inputenc}
\usepackage[T1]{fontenc}
\usepackage{fullpage}

\newcommand{\id}{\mathsf{id}}
\newcommand{\bad}{\mathsf{B}}
\newcommand{\CONGEST}{\textsc{congest}}

\newtheorem{ourclaim}{Claim}
\newtheorem{theorem}{Theorem}
\newtheorem{lemma}{Lemma}
\newtheorem{proposition}{Proposition}
\newenvironment{claimproof}[1]{\par\noindent\emph{Proof of claim.}\space#1}{\hfill$\diamond$\medskip}

\title{On the Power of Threshold-Based Algorithms for Detecting Cycles in the \textsc{CONGEST} Model}

\date{\vspace{-5ex}}

\author{
Pierre Fraigniaud$^1$, Maël Luce$^1$, and Ioan Todinca$^2$
}

\begin{document}

\maketitle

\begin{center}

\vspace{5mm}

$^1$Institut de Recherche en Informatique Fondamentale (IRIF) \\ CNRS and Université Paris Cité, France \\ 
$^2$Laboratoire d'Informatique Fondamentale d'Orléans (LIFO) \\ Université d'Orléans, France
\end{center}


\begin{abstract}
It is known that, for every $k\geq 2$, $C_{2k}$-freeness can be decided by a generic Monte-Carlo algorithm running in $n^{1-1/\Theta(k^2)}$ rounds in the \CONGEST\/ model. For $2\leq k\leq 5$, faster Monte-Carlo algorithms do exist, running in $O(n^{1-1/k})$ rounds, based on upper bounding the number of messages to be forwarded, and aborting search sub-routines for which this number exceeds certain thresholds. We investigate the possible extension of these \emph{threshold-based} algorithms, for the detection of larger cycles. We first show that, for every $k\geq 6$, there exists an infinite family of graphs containing a $2k$-cycle for which any threshold-based algorithm fails to detect that cycle. Hence, in particular, neither $C_{12}$-freeness nor $C_{14}$-freeness can be decided by threshold-based algorithms. Nevertheless, we show that $\{C_{12},C_{14}\}$-freeness can still be decided by a threshold-based algorithm, running in $O(n^{1-1/7})= O(n^{0.857\dots})$ rounds, which is faster than using the generic algorithm, which would run in $O(n^{1-1/22})\simeq O(n^{0.954\dots})$ rounds. Moreover, we exhibit an infinite collection of families of cycles such that threshold-based algorithms can decide $\mathcal{F}$-freeness for every $\mathcal{F}$ in this collection. 

\end{abstract}


\section{Introduction}

\subsection{Objective}

Graphs excluding a fixed family $\mathcal{F}$ of graphs, whether it be as subgraphs, induced subgraphs, topological subgraphs, or minors, play a huge role in theoretical computer science, especially in graph theory as well as in algorithm design and complexity, from standard and parametrized complexity, to the design of approximation and exact algorithms. Famous examples in structural graph theory are \emph{Wagner's theorem} stating that  a finite graph is planar if and only if it does not have $K_5$ or $K_{3,3}$ as a minor,  and the  \emph{forbidden subgraph problem} which looks for the maximum number of edges in any $n$-vertex graph excluding a given graph~$G$ as induced subgraph. In the algorithm and complexity framework, it is known that the vertex coloring problem is NP-hard in triangle-free (i.e., $C_3$-free) graphs, but many families $\mathcal{F}$ have been identified, for which computing the chromatic number of graphs excluding every graph in~$\mathcal{F}$ as induced subgraphs can be done in polynomial time. For instance, it is known that, for a graph $H$ of at most six vertices, vertex coloring for $\{C_3,H\}$-free graphs is
polynomial-time solvable if $H$ is a forest not isomorphic to $K_{1,5}$, and NP-hard otherwise~\cite{BroersmaGPS12}. Another recent illustration of the importance of $\mathcal{F}$-free graphs, is graph   isomorphism, which can be tested in time $n^{\mbox{\scriptsize polylog}(k)}$ on all $n$-node graphs excluding an arbitrary $k$-node graph as a topological subgraph~\cite{Neuen2022}. 

In the context of distributed computing for networks however, still very little is known about $\mathcal{F}$-free graphs, even for the most basic case where the graphs in~$\mathcal{F}$ must be excluded as mere subgraphs (not necessarily induced). In fact, up to our knowledge, most of the work in this domain has focused on the standard \CONGEST\/ model, and its variants. Recall that the \CONGEST\/ model is a distributed computing model for networks where the nodes of a graph execute the same algorithm, as a sequence of synchronous rounds, during which every node is bounded to exchange messages of $O(\log n)$ bits with each of its neighbors (see~\cite{Peleg2000}). Also recall that a distributed algorithm $\mathcal{A}$  \emph{decides} a graph property $P$ if, for every input graph~$G$, the following holds: $G$ satisfies $P$ if and only if $\mathcal{A}$ accepts at every node of $G$. Deciding $H$-freeness is a fruitful playground for inventing new techniques for the design of efficient \CONGEST\/ algorithms. Indeed, the problem itself is \emph{local}, yet the limited bandwidth of the links imposes severe limitations on the ability of every node to gather information about nodes at distance more than one from it. 

An important case is checking the absence of a cycle of given size as a subgraph. On the negative side, for every $k\geq 2$, deciding $C_{2k+1}$-freeness requires $\tilde{\Omega}(n)$ rounds in \CONGEST, even for randomized algorithms~\cite{drucker2014}. However, for every $k\geq 2$, $C_{2k}$-freeness can be solved by Monte-Carlo algorithms performing in a \emph{sub-linear} number of rounds. For instance $C_4$-freeness can be decided (deterministically) in $O(\sqrt{n})$ rounds~\cite{drucker2014}, and, for every $k\geq 3$, the round-complexity of $C_{2k}$-freeness is at most 
$O(n^{1-2/(k^2-2k+4)})$ if $k$ is even, and 
$O(n^{1-2/(k^2-k+2)})$ if $k$ is odd (see~\cite{fischer2018}).

 The round-complexity of deciding $C_{2k}$-freeness has been recently improved (see \cite{Censor-HillelFG20}), for small values of~$k$, by an elegant algorithm which, for every $2\leq k \leq 5$, runs in   $O(n^{1-1/k})$ rounds. For $k=2$ the (randomized)  algorithms in~\cite{Censor-HillelFG20,fischer2018} runs with the same asymptotic complexity as the (deterministic) algorithm in~\cite{drucker2014}, i.e., in $O(\sqrt{n})$ rounds, and this cannot be improved, up to a  logarithmic  multiplicative factor~\cite{drucker2014}. However, for $k\in\{3,4,5\}$, the current best-known upper bound on the round-complexity of deciding $C_{2k}$-freeness is   $O(n^{1-1/k})$.  Interestingly, the algorithm in~\cite{Censor-HillelFG20} also allows to decide whether the girth of a network is at most~$g$,  in $\tilde{O}(n^{1-2/g})$ rounds. In other words, the algorithm decides  $\{C_k,\,3\leq k \leq g\}$-freeness for any given~$g$.

In a nutshell, the algorithm in~\cite{Censor-HillelFG20} is based on the notion of \emph{light} and \emph{heavy} nodes, where a node is light if its degree is at most $n^{1/k}$, and heavy otherwise. Cycles of length $2k$ composed of light nodes only can be found in at most $\sum_{i=0}^{k-1}n^{i/k}=\Theta(n^{1-1/k})$ rounds, by brute-force search, using color-coding~\cite{AlonYZ95}. For finding cycles containing at least one heavy node, it is noticed that, by picking a node~$s$ uniformly at random,  the probability that $s$ is neighbor of a heavy node is at least $n^{1/k}/n$, and thus, by repeating the experiment $\Theta(n^{1-1/k})$ times, a neighbor of a heavy node belonging to some $2k$-cycle will be found with constant probability, if it exists. The node~$s$ chosen at a given time of the algorithm initiates  brute-force searches from all its heavy neighbors in parallel, each one searching for a cycle containing it, using color coding.  The main point in the algorithm is the following. It is proved that, for every $k\in\{2,3,4,5\}$, and every $i\in\{1,\dots,k-1\}$, there is a \emph{constant} threshold~$T_k(i)$ such that, if a node colored~$i$ or $2k-i$ has to forward more than $T_k(i)$ searches initiated from the heavy neighbors of~$s$, then that node can safely abort the search, without preventing the algorithm from eventually detecting a $2k$-cycle, if it exists. It follows that the parallel searches initiated by the random source~$s$ run in $O(1)$-rounds, and thus the ``threshold-based'' algorithm in~\cite{Censor-HillelFG20}  runs in $O(n^{1-1/k})$ rounds overall. 

The objective of this paper is to determine under which condition, and for which graph family~$\mathcal{F}$,  threshold-based algorithms can be used for deciding $\mathcal{F}$-freeness.

\subsection{Our Results}

Our first contribution is a negative result.  For every $k\geq 6$, we exhibit an infinite family of graphs in which any threshold-based algorithm fails to decide $C_{2k}$-freeness.  That is, we show that, for $k\geq 6$, a threshold-based algorithm must forward a non-constant amount of messages at some step to guarantee that the parallel searches initiated by the random source $s$ detect a $2k$-cycle. More specifically, we show the following.

\begin{theorem}\label{theo:does-not-extend}
For every $k\geq 6$, there exists an infinite family $\mathcal{G}$ of graphs containing a unique $2k$-cycle $C=(u_0,u_1,\dots,u_{2k-1})$ such that, for every $T \in o(n^{1/6}/\log n)$, the threshold-based algorithm fails to detect~$C$ in at least one $n$-node graph in~$\mathcal{G}$ if the thresholds are set to~$T$. 
\end{theorem}

In other words, Theorem~\ref{theo:does-not-extend} says that, for every $k\geq 6$,  there are no efficient threshold-based algorithms capable to decide $C_{2k}$-freeness. In particular, neither $C_{12}$-freeness nor $C_{14}$-freeness can be decided by a threshold-based algorithm. Nevertheless, our second contribution states that this is not the case of determining whether a graph is free of both $C_{12}$ and $C_{14}$.

\begin{theorem}\label{theo:small-families}
$\{C_{12},C_{14}\}$-freeness can be decided by a threshold-based algorithm running in  $O(n^{1-\frac{1}{7}})$ rounds.
\end{theorem}

Note that the generic algorithm from~\cite{fischer2018} would run in $O(n^{1-\frac{1}{22}})= O(n^{0.954\dots})$ rounds for deciding $\{C_{12},C_{14}\}$-freeness by checking separately whether the graph contains a $C_{12}$, and whether the graph contains a $C_{14}$. Instead, our algorithm performs in $O(n^{1-1/7})= O(n^{0.857\dots})$ rounds. Note that establishing that $\{C_{10},C_{12}\}$-freeness can be decided by a threshold-based algorithm running in  $O(n^{1-\frac{1}{6}})$ rounds is rather easy because $C_{10}$-freeness can be decided by such an algorithm. The point is that, again, thanks to Theorem~\ref{theo:does-not-extend}, neither  $C_{12}$-freeness nor  $C_{14}$-freeness can be decided by a threshold-based algorithm. 

Finally, note that, by construction, threshold-based algorithms can decide $\mathcal{F}_k$-freeness, for every $k\geq 2$, where $\mathcal{F}_k={\{C_{2\ell}\mid 2\leq \ell \leq k\}}$. This raises the question of identifying infinite collections of smaller families $\mathcal{F}$ of cycles for which threshold-based algorithms succeed to decide $\mathcal{F}$-freeness. We identify two such families. 

\begin{theorem}\label{theo:one-every-four}
Let $k\geq 2$,  $\mathcal{F}'_k=\{C_{4\ell}\mid 1\leq \ell \leq k\}$, and ${\mathcal{F}''_k=\{C_{4\ell+2} \mid 1\leq \ell \leq k\}}$. Both $\mathcal{F}'_k$-freeness and $\mathcal{F}''_k$-freeness can be  decided by threshold-based algorithms running in $\tilde{O}(n^{1-1/2k})$ rounds, and $\tilde{O}(n^{1-1/(2k+1)})$ rounds, respectively.
\end{theorem}

Due to lack of space, the proof of this latter theorem in placed in Appendix~\ref{app:theo:one-every-four}.


\subsection{Related Work}

Deciding $H$-freeness for a given graph~$H$ has been considered in~\cite{eden2022}, which describes an algorithm running in $\tilde{O}(n^{2-2/(3k+1)+o(1)})$ rounds for $k$-node graphs~$H$. This round complexity is nearly matching the general lower bound $\tilde{\Omega}(n^{2-\Theta(1/k)})$ established in~\cite{fischer2018}. This latter bound can be overcome for specific graphs~$H$, and typically when $H$ is a cycle. 

For every $k\geq 3$, $C_k$-freeness can be decided in $O(n)$ rounds (see, e.g., \cite{korhonen2017}). However, the exact round-complexity of deciding $C_k$-freeness varies a lot depending on whether $k$ is even or odd. It was proved in~\cite{drucker2014} that deciding $C_{2k+1}$-freeness requires $\tilde{\Omega}(n)$ rounds for $k\geq 2$. Nevertheless, sub-linear algorithms are known for even cycles. In particular, the round-complexity of deciding $C_4$-freeness was established as $\tilde{\Theta}(n^{1/2})$ in~\cite{drucker2014}. The lower bound $\tilde{\Omega}(n^{1/2})$ rounds is also known to hold for deciding $C_{2k}$-freeness, for every $k\geq 3$~\cite{korhonen2017}. The best generic upper bound for deciding $C_{2k}$-freeness is $\tilde{O}(n^{1-\Theta(1/k^2)})$ rounds~\cite{fischer2018}. Faster algorithms are known, but for specific values of $k$ only. Specifically,  for every $k\in \{3,4,5\}$, $C_{2k}$-freeness can be decided in  $O(n^{1-1/k})$ rounds~\cite{Censor-HillelFG20}. The special case of triangle detection, i.e., deciding $C_3$-freeness is widely open. 

It may also be worth mentioning the study of cycle-detection in the context of a model stronger than \CONGEST, namely in the \textsc{congested clique} model. In this model, efficient algorithms have been designed. In particular, it was shown in~\cite{censor2015} that $C_3$-freeness can be decided in $O(n^{0.158})$ rounds, $C_4$-freeness can be decided in~$O(1)$ rounds, and $C_k$-freeness can be decided in $O(n^{0.158})$ rounds for any $k\geq 5$. 

The interested reader is referred to~\cite{censor2022} for a recent survey on subgraph detection, and related problems, in  \CONGEST\/ or  similar models.

%
%


\section{Preliminaries}
\label{sec:preliminaries}

In this section, we recall the main techniques used for deciding whether the graph contain a cycle of a given length as a subgraph, and we summarize the threshold-based algorithms defined in~\cite{Censor-HillelFG20}. 

\subsection{Subgraph Detection} 

Recall that a graph $H$ is a \emph{subgraph} of a graph $G$ if $V(H)\subseteq V(G)$ and $E(H)\subseteq E(G)$. Given a graph~$H$, a deterministic distributed algorithm $\mathcal{A}$ for the CONGEST model decides $H$-freeness in $R(\cdot)$ rounds if, for every $n$-node graph~$G$, whenever $\mathcal{A}$ runs in~$G$, each node outputs ``accept'' or ``reject'' after $R(n)$ rounds, and 
\[
\mbox{$G$ contains $H$ as a subgraph} \iff \mbox{at least one node of $G$ rejects.}
\]
A randomized Monte-Carlo algorithm decides $H$-freeness if, for every $n$-node graph~$G$, 
\[
\left\{\begin{array}{l}
\mbox{$G$ contains $H$ as a subgraph}  \Longrightarrow  \Pr[\mbox{at least one node of $G$ rejects}]\geq 2/3.\\
\mbox{$G$ does not contain $H$ as a subgraph}  \Longrightarrow  \Pr[\mbox{all nodes accept}]\geq 2/3.
\end{array}\right.
\]
In fact, most algorithms for deciding $H$-freeness are 1-sided, i.e., they alway accept $H$-free graphs, and may err only by failing to detect an existing copy of $H$ in~$G$. By repeating the execution of 1-sided error algorithms for sufficiently many times, one can make the error probability as small as desired. 

In the case $H=C_{2k}$, which is the framework of this paper, one standard technique, called \emph{color-coding}~\cite{AlonYZ95}, plays a crucial role and was used in many algorithms for detecting cycles in various contexts (see, e.g., \cite{Censor-HillelFG20,EvenFFGLMMOORT17,FraigniaudO19}). 

\paragraph{Color Coding.} 

Let $G=(V,E)$, and $W\subseteq V$. For deciding whether there is a $2k$-cycle including one node in $W$, let every node of $G$ pick a color in $\{0,\dots,2k-1\}$ uniformly at random. Then every node $w\in W$ colored~0 launches a search, called $\textsf{color-BFS}(k,w)$,  by sending its identifier to all its neighbors colored~1 and~$2k-1$. Every node colored~1 receiving an identifier from a node colored~0 forwards that identifier to all its neighbors colored~2, while every node colored~$2k-1$ receiving an identifier from a node colored~0 forwards it to all its neighbors colored~$2k-2$. More generally, for every $i=2,\dots,k-1$, every node colored~$i$ receiving an identifier from a node colored~$i-1$ forwards it to all its neighbors colored~$i+1$, and, for every $i=2k-2,\dots,k+1$, every node colored~$i$ receiving an identifier from a node colored~$i+1$ forwards it to all its neighbors colored~$i-1$. If a node colored~$k$ receives a same identifier from a neighbor colored~$k-1$, and from a neighbor colored~$k+1$, then it rejects. 

The number of rounds required by $\textsf{color-BFS}(k,W)$ is at most $k\, |W|$. Also, if the nodes in the graphs have maximum degree~$\Delta$, then the number of rounds is at most~$O(\Delta^{k-1})$. Overall, we have 
\begin{equation}\label{eq:colBFS}
\mbox{\#rounds \sf color-BFS}(k,W)=O(\min\{k\,|W|,\; \Delta^{k-1}\}).
\end{equation}
If there is a $2k$-cycle passing through a node in~$W$, then the probability that this cycle is colored appropriately is at least $\rho=1/(2k)^{2k}$, and therefore the cycle is found with probability at least~$\rho$. By repeating the procedure a constant number of times proportional to $(2k)^{2k}$, the cycle is found with probability at least~$2/3$.  


\subsection{Threshold-Based Algorithms}
\label{subsec:threshold-based-algo}

We denote the algorithm defined in~\cite{Censor-HillelFG20} by  $\mathcal{A}^\star$. This algorithm, summarized in Algorithm~\ref{algo}, heavily uses color-coding. A node~$u$ of $G$ is called \emph{light} if $\deg(u)\leq n^{1/k}$, and \emph{heavy} otherwise. A $2k$-cycle $C$ containing only light nodes is called \emph{light cycle}, and is heavy otherwise. 

\paragraph{Detecting light cycles.} 

Detecting whether there is a light $2k$-cycle is easy by applying color-coding in the subgraph $G[U]$ of $G$ induced by light nodes~$U$ (i.e.,  only light nodes participate). By Eq.~\eqref{eq:colBFS} with $\Delta=n^{1/k}$, we get that the detection of light cycles  takes $O(n^{1-1/k})$ rounds. If a light $2k$-cycle exists in the graph, some light node rejects with constant probability, and we are done. 

\paragraph{Detecting heavy cycles.} 

For detecting heavy cycles, $\mathcal{A}^\star$ picks a node $s$ uniformly at random in the graph\footnote{This can be done by letting each node choosing an integer value in $\{1,\dots,n^3\}$ uniformly at random; The node $s$ with smallest value is the chosen node. With high probability, this node is unique.}. The idea is that if there is a heavy $2k$-cycle in the graph, say $C=(u_0,u_1,...,u_{2k-1})$ where $u_0$ is heavy, then the probability that a neighbor $s$ of $u_0$ is picked is at least $n^{-(1-1/k)}$ since $\deg(u_0)\geq n^{1/k}$. Therefore, by repeating $\Theta(n^{1-1/k})$ times the choice of~$s$, a neighbor of $u_0$ will be picked with constant probability. For each choice of~$s$, the goal is to proceed with searching a $2k$-cycle in a constant number of rounds.  

The chosen node~$s$ launches $\textsf{color-BFS}(k,s)$ for figuring out whether there is a $2k$-cycle passing through~$s$. By Eq.~\eqref{eq:colBFS} with $|W|=1$, this takes $O(1)$ rounds. If a $2k$-cycle is detected, some node rejects, and we are done. 

We therefore assume from now that $s$ does not belong to a $2k$-cycle. The source node $s$ then sends a message to all its heavy neighbors~$W$, and each of these neighbors $w$ launches $\textsf{color-BFS}(k,w)$, in parallel. At this point, one cannot simply rely on Eq.~\eqref{eq:colBFS} with $|W|\leq\deg(s)$ to bound the round-complexity of $\textsf{color-BFS}(k,W)$ because $s$ may have non-constant degree. The central trick used in~\cite{Censor-HillelFG20} consists to provide each node with a threshold for the number of messages the node can forward at a given step of a $\textsf{color-BFS}$. In case the number of messages to be transmitted exceeds the threshold, then the node aborts, i.e., it stops participating to the current $\textsf{color-BFS}$. It is shown that such threshold-based approach may prevent the nodes to detect $2k$-cycles, but not too often, and that a $2k$-cycle will be detected with constant probability anyway, if it exists. This is summarized by the following lemma. 

\begin{lemma}[\cite{Censor-HillelFG20}]\label{lem:de-Censor-Hillel-et-al}
Let $C=(u_0,u_1,...,u_{2k-1})$ be a $2k$-cycle in~$G$, with $u_0$ heavy, and of maximum degree among the nodes in~$C$. For every $k\in\{2,3,4,5\}$, there exists a constant $\alpha_k>0$, and there exist constant thresholds $T_k(i)$, ${i=1,\dots,k-1}$, such that, even if nodes colored $i$ or $2k-i$ abort the search launched from the set~$W$ of heavy neighbors of~$s\in N_{G}(u_0)$ at the $i$-th step of $\textrm{\sf color-BFS}(k,W)$ whenever they generate a congestion larger than $T_k(i)$, still, for a fraction at least $\alpha_k$ of the neighbors~$s$ of~$u_0$,  the cycle~$C$ will be found, unless $s$  itself belongs to a $2k$-cycle. 
\end{lemma}

That is, $\mathcal{A}^\star$ sets thresholds (depending on~$k$), and if the volume of communication generated by the $\textsf{color-BFS}(k,W)$ launched in parallel by all the heavy neighbors $W$ of a random source~$s$ exceeds these thresholds, then the search aborts. Yet,  it is proved in~\cite{Censor-HillelFG20} that this does not prevent a $2k$-cycle to be found, if it exists. 

\begin{algorithm}
\caption{Deciding $C_{2k}$-Freeness by the Threshold Algorithm $\mathcal{A}^\star$ from \cite{Censor-HillelFG20}}
\label{algo}
\begin{algorithmic}[1]
\State $\textsf{color-BFS}(k,U)$ in $G[U]$ \Comment{$U=\{u\in V(G)\mid \deg(u)\leq n^{1/k}\}$} \label{inst:1}
\For {$i=1$ \textbf{to} $\Theta(n^{1-1/k})$}
\State $s\leftarrow$ random node in $G$ \Comment{$W=\{v\in N_G(s)\mid \deg(v)> n^{1/k}\}$}
\State $\textsf{color-BFS}(k,s)$ \label{inst:4}
\State $\textsf{color-BFS}(k,W)$ with threshold $T_k(i), i\in\{1,\dots,k-1\}$ \label{inst:5} 
\EndFor
\end{algorithmic}
\end{algorithm}

The algorithm $\mathcal{A}^\star$ is summarized in Algorithm~\ref{algo}. Each \textsf{color-BFS} includes $\sim (2k)^{2k}$ executions of color-coding to guarantee $2k$-cycle detection with probability at least~$2/3$. Instruction~\ref{inst:1} performs in $\Theta(n^{1-1/k})$ rounds because $G[U]$ has maximum degree~$n^{1/k}$. It finds a light $2k$-cycle, if it exists, with probability~$2/3$.  Instruction~\ref{inst:4} performs in $O(1)$ rounds for each constant~$k\leq 5$, and, if $s$ belongs to a $2k$-cycle, it finds such a cycle with probability~$2/3$. Instruction~\ref{inst:5} also performs in $O(1)$ rounds as well, thanks to the thresholds specified in Lemma~\ref{lem:de-Censor-Hillel-et-al}. If there is a heavy $2k$-cycle, and if $s$ does not belong to a $2k$-cycle, then that heavy $2k$-cycle is found with probability~$2/3$. Overall, $\mathcal{A}^\star$ performs in $O(n^{1-1/k})$ rounds, and succeeds with probability~$2/3$.

In the next section, we shall show that thresholds $T_k(i)$, ${i=1,\dots,k-1}$, such as the ones specified in  Lemma~\ref{lem:de-Censor-Hillel-et-al} cannot be set for $k\geq 6$.


\section{Limits of the Threshold-Based Algorithms}

This section is entirely dedicated to the proof of Theorem~\ref{theo:does-not-extend}, which is essentially based on proving the impossibility of setting a constant $T_k(k-3)$ for $k \geq 6$. For this purpose, we exhibit a class of graph $\{G_k \mid k\geq 6\}$ such that each $G_k$ does not contain any light cycle~$C_{2k}$, and contains exactly one heavy cycle~$C_{2k}$. The construction of $G_k$ for $k \geq 6$ is split in two cases: a generic construction, which works for all $k\geq 7$, and a specific construction for~$G_6$. We begin the proof by the generic case. 


Let $k\geq 7$.  The graph $G_k$ is composed of the following nodes (see Fig.~\ref{Gforkatleast7}), for $N\geq 1$:
\begin{itemize}
\item The $2k$ nodes of the unique $2k$-cycle $C^\star=(u_0,u_1,...,u_{2k-1})$;
\item The set $S=\{s^p, \, p\in\{1,\dots,N\}\}$ of  $N$ neighbors of $u_0$;
\item The set $W=\{w_{k-4}^q, \, q\in\{1,\dots,N\}\}$ of  $N$ neighbors of $u_{k-3}$;
\item For $(p,q)\in\{1,\dots,N\}^2$, the set $\{w_j^{p,q}, \; j\in\{0,\dots,k-5\}\}$ of the nodes on a path from node $s^p$ to node in~$w^q_{k-4}$;
\item For $(p,q)\in\{1,\dots,N\}^2$, the set $\{v_0^{p,q,r}, \; r\in\{1,\dots,N\}\}$ of private neighbors of node~$w_0^{p,q}$ (these nodes are added in order to ensure that $w_0^{p,q}$ is heavy);
\item For $(p,q)\in\{1,\dots,N\}^2$, the set $\{v_{k-5}^{p,q,r} : r\in\{1,\dots,N\}\}$ of private neighbors of node~$w_{k-5}^{p,q}$ (as above, this makes node $w_{k-5}^{p,q}$ heavy).
\end{itemize}
The number of nodes in $G_k$ is $n=\Theta(N^3)$.  

\begin{figure}
\centering
\makebox[\textwidth][c]{\includegraphics[scale=0.8]{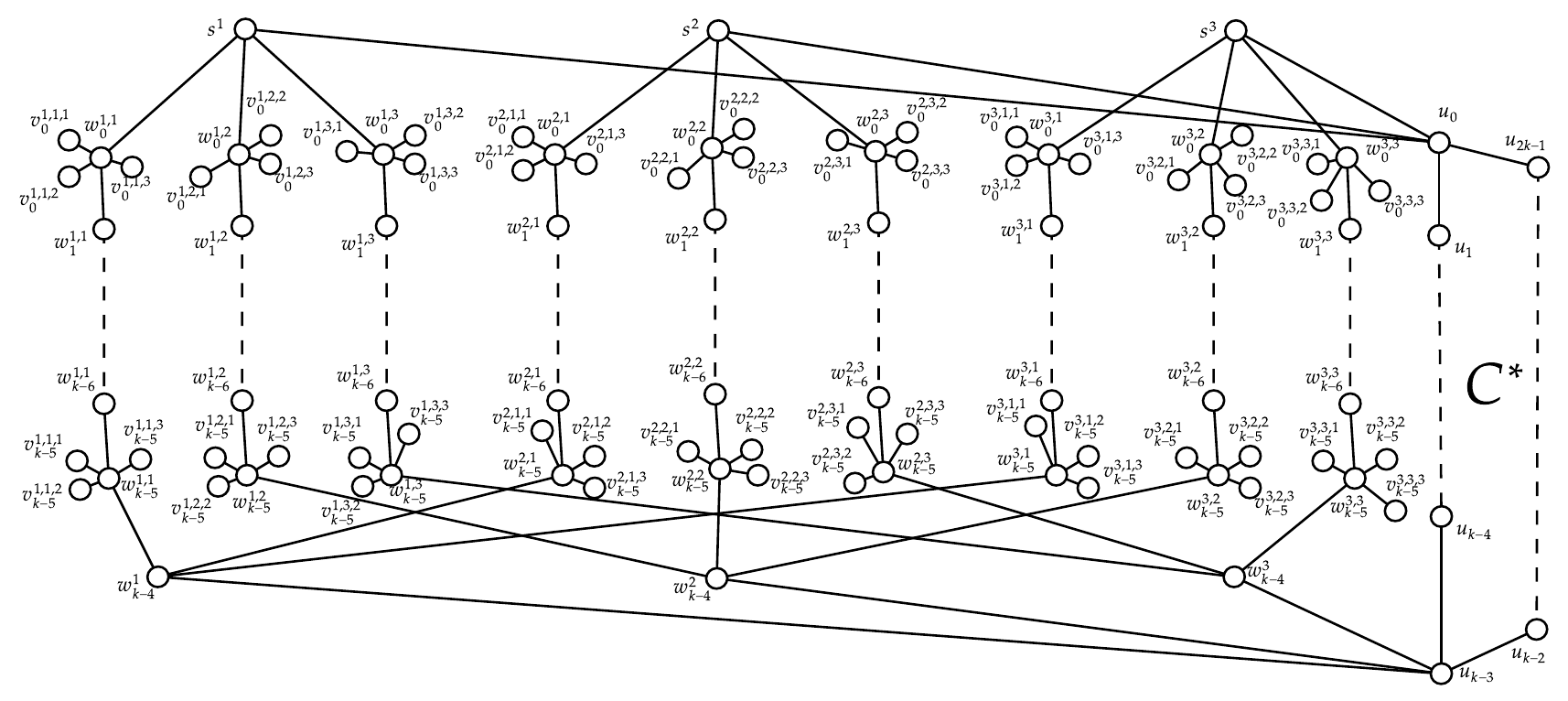}}
\caption{The graph $G_k$ for $k\geq 7$ and $N=3$.}
\label{Gforkatleast7}
\end{figure}

The proof of the following result can be found in Appendix~\ref{app:lem:C*unique}. 

\begin{lemma}\label{lem:C*unique}
For every $k\geq 7$, $C^\star$ is the unique $2k$-cycle in $G_k$, and is a heavy cycle.
\end{lemma}

As a consequence of Lemma~\ref{lem:C*unique}, a $2k$-cycle in~$G_k$ can only be detected if the algorithm picks the random source~$s$ in $N_{G_k}(C^\star)$, i.e., it must pick $s \in S\cup W\cup C^\star$. Also, if $s\in S\cup W$, then $s$ does not belong to a $2k$-cycle, and thus $s$ will initiate the search for $C_{2k}$ from each of its heavy neighbors. 

\begin{lemma}\label{lem:lowsuccessfork7}
Let $T\in o(n^{1/3}/\log n)$, and let us set ${T_k(k-3)=T}$ in the threshold-based algorithm. If $s\in S$ (resp., $s\in W$), then the probability that~$u_{k-3}$ (resp.,~$u_0$) forwards at most $T$ messages during a search phase from heavy nodes is $\exp(-\Theta(n^{1/3}))$.
\end{lemma}

\begin{proof}
By the symmetry of~$G_k$, the roles of $S\cup\{u_0\}$ and $W\cup\{u_{k-3}\}$ are identical. We shall thus prove the lemma only for $s\in S$, i.e., $s=s^p$ for some $p\in\{1,\dots,N\}$. For every $q\in \{1,\dots,N\}$, let $X_q$ be the following Bernoulli random variable, assuming each node picks a color in $\{0,\dots,2k-1\}$ uniformly at random. We say that the path $w_0^{p,q},\dots,w_{k-5}^{p,q},w_{k-4}^q$ is well-colored if, for every $i\in\{0,\dots,k-5\}$, $w_i^{p,q}$ is colored~$i$, and $w_{k-4}^q$ is colored $k-4$. We define
\[
X_q = \left\{\begin{array}{ll}
        1 & \mbox{if the path }(w_0^{p,q},\dots,w_{k-5}^{p,q},w_{k-4}^q)\textrm{ is well-colored;}\\
        0 & \mbox{otherwise.}
    \end{array}\right.
\]
Let then $X = \sum_{q=1}^N X_q$ be the random variable that counts the number of identifiers different from $\id(u_0)$ that $u_{k-3}$ has to forward, that is, the identifiers of all the nodes $w_0^{p,q}$ satisfying $X_q=1$.
$X$ follows a Binomial law of parameters $N$ and $r = (\frac{1}{2k})^{k-3}$, so its expectation is $E(X) = N r$.
Since every node $w_0^{p,q}$ has a unique path with length $k-3$ to node $u_{k-3}$, we get that 

\begin{align*}
\Pr[X \leq T] &= \sum_{t=0}^{T}\Pr[X=t] =\sum_{t=0}^{T}\binom{N}{t}r^t(1-r)^{N-t}\\
&\leq  \sum_{t=0}^{T}N^t\, r^t e^{(N-t)\ln(1-r)} \leq  (T+1)\, N^T \, r^T \, e^{N \ln(1-r)}\\
&= N^Te^{-\Theta(N)}\textrm{.}
\end{align*}

Therefore, the probability that $u_{k-3}$ has to forward at most $T$ messages is $O(N^Te^{-\Theta(N)})$. If $T=o(N/\log N)\simeq o(n^{1/3}/\log n)$, then this probability is asymptotically equal to $\exp(-\Theta(n^{1/3}))$. 
\end{proof}

To conclude the proof of Theorem~\ref{theo:does-not-extend} for $k\geq 7$ note that, even by fixing all thresholds to $T\in o(n^{1/3}/\log n)$, Algorithm $\mathcal{A}^\star$ fails to detect the unique (heavy) cycle $C^\star$ almost surely. Indeed when picking vertex $s$ in $S$ (or, symmetrically, $s$ in $W$), the algorithm succeeds with probability $\exp(-\Theta(n^{1/3}))$, since vertex $u_{k-3}$ aborts almost surely. The other possibility of detecting the cycle is when the algorithm picks $s$ directly on $C^\star$, which only contains $2k$ vertices, so the success possibility is $O(1/n)$. Hence, although the algorithm makes $\tilde O(n^{1-1/k})$ independent random choices of $s$, the probability of success is only $\tilde{O}(n^{-1/k})$.


The specific case $k=6$ is treated in Appendix~\ref{app:specific-for-k6}, which completes  the proof of Theorem~\ref{theo:does-not-extend}.


\section{Deciding $\{C_{12},C_{14}\}$-Freeness}

This section is entirely dedicated to the proof of Theorem~\ref{theo:small-families}. 
We rely mostly on the threshold algorithm as such, with the following slight modification, for simplifying the analysis. 

\paragraph{\textit{Remark.}}

For exhibiting the thresholds $T_{2k}(i)$, $1\leq i \leq 2k-1$, it is convenient to assume that, instead of repeating $O(n^{1-1/k})$ random choices of~$s$, and then, for each chosen~$s$, repeating $\sim (2k)^{2k}$ random choices of colors (for the color-BFSs), the algorithm proceeds as follows: The outer loop repeats $\sim (2k)^{2k}$ random assignments of colors, and the inner loop repeats $O(n^{1-1/k})$ random choices of~$s$ (for each of the $\sim (2k)^{2k}$  color-assignments). In fact, it  simplifies the presentation even further by assuming that the random colors are in the range ${\{-1,0,\dots,2k-1\}}$. The extra color~$-1$ is used only by~$s$, and $s$ launches $\textsf{color-BFS}(W)$ only under the condition that $s$ has random color~$-1$. None of these changes affect the performances of the algorithm, up to a constant factor in the round-complexity. 

\bigbreak 

The algorithm starts by checking the existence of a light 12-cycle or a light 14-cycle. This is achieved in $O(n^{6/7})$ rounds, by parallel color-BFSs running on the light nodes only (see Section~\ref{subsec:threshold-based-algo}). For detecting heavy cycles, the algorithm proceeds as the threshold algorithm, by repeating the choice of a random node~$s$. For each choice, the chosen node~$s$ checks whether it belongs to a  12-cycle or to a 14-cycle, by performing two series of $\textsf{color-BFS}(s)$, one for detecting a possible 12-cycle passing through~$s$, and one for detecting a  possible 14-cycle. If no such cycles are detected, then $s$ proceeds as follows. 

\begin{description}
\item[Looking for 14-cycles.] Node $s$ launches $\textsf{color-BFS}(W)$, from the set $W$ of all its heavy neighbors, with appropriate thresholds  $T_7(i)$, $1\leq i \leq 6$, that will be specified later. The crucial point here is that if the algorithm proceeds by checking the existence of 12-cycles \emph{and} of 14-cycles. By checking both lengths, we will be able to establish a result similar to Lemma~\ref{lem:de-Censor-Hillel-et-al}, that is, if there is a $14$-cycle in~$G$, say $C=(u_0,\dots,u_{13})$, with $u_0$ heavy,  and of maximum degree among the nodes in~$C$, then there exists a constant $\alpha>0$, such that, even if nodes colored $i$ or $14-i$ abort the search launched from the set~$W$ at the $i$-th step of $\textrm{\sf color-BFS}(W)$ whenever they generate a congestion larger than $T_7(i)$, still, for a fraction at least $\alpha$ of the neighbors~$s$ of~$u_0$,  the cycle~$C$ will be found with probability at least~$2/3$. In other words, if a node rejects during this phase, it is because there is a 12-cycle, or there is a 14-cycle. On the other hand, the fact that all nodes accept during this phase only provides a (statistical) guarantee on the absence of 14-cycles, but provides little information on the absence of 12-cycles. 
  
\item[Looking for 12-cycles.] Again, node $s$ launches $\textsf{color-BFS}(W)$, but for 12-cycles now, with the mere thresholds  $T_6(i)=1$ for all $i=1,\dots,5$. The crucial point here is that, assuming that the graph is $C_{14}$-free,  then a threshold of~1 suffices. There will only ever be one message crossing an edge in a well-colored heavy 12-cycle. This latter fact is easy to establish, so most of the proof consists in proving the existence of the thresholds when looking for 14-cycles. 
\end{description}

\bigbreak

To set the values $T_7(i)$ for $1\leq i\leq 6$, let us define $T_7(0)=1$. Our construction is then inductive, and, for $i>0$, we shall set  $T_7(i) = f(i) \cdot T_7(i-1)$ for appropriate constants $f(i)$. 
Let us assume that the graph contains a 14-cycle, denoted by $C^\star=(u_0,u_1,\dots,u_{13})$, where $u_0$ is of  maximum degree in~$C^\star$, and, for every $i=0,\dots,13$, node~$u_i$ is colored~$i$. 
From now on, we will work only on the nodes $u_0,u_1,\dots,u_7$. By symmetry, the same arguments will apply to nodes $u_0,u_{13},\dots,u_7$. Before further defining the setting of the proof, recall that, as underlined before, the $O(n^{1-1/k})$ drawings of nodes $s$ are performed on a given coloring of the graph with colors in $\{-1,0,\dots,13\}$, and that only a picked node~$s$ colored~$-1$ invokes color-BFS($W$). 

The lemma below is generic, as it applies to all $k\geq 2$. Recall that a  path  $s,w_0,\dots,w_{i-1}, u_i$ from node~$s$ to node~$u_i$ is \emph{well-colored} if $s$ is colored~$-1$, $u_i$ is colored~$i$, and, for every $j=0,\dots,i-1$, node $w_j$ is colored~$j$ by color-coding.

\begin{lemma}\label{lem-paths-from-single-s}
Let $k\geq 2$ be an integer. For every  $i\in \{1,\dots,k-1\}$, let $\rho$ be the maximum number of node-disjoint well-colored paths from $s$ to~$u_i$. If $s$ launches $\textsf{color-BFS}(W)$ from all its heavy neighbors colored~$0$, then $u_i$ cannot receive more than $\rho\cdot T_k(i-1)$ identifiers from nodes colored~$i-1$. 
\end{lemma}

\begin{proof}
Let $S$ be a set of $\rho$ node-disjoint, well-colored paths from $s$ to~$u_i$.
Let $w_0\in W$ be a heavy neighbor of $s$ colored~0, and let us assume that $\id(w_0)$ has reached~$u_i$. It follows that there is a well-colored path~$P$ of length $i$ from $w_0$ to $u_i$. The path~$P$ must intersect some path $P'=\{w'_0,\dots,w'_{i-1}\}$ in~$S$ (perhaps even $P=P'$). As a consequence, $\id(w_0)$ is included in the at most $T_k(i-1)$ identifiers that node $w'_{i-1}$ may forward to~$u_i$. Therefore, the number of identifiers received by $u_i$ during $\textsf{color-BFS}(W)$ does not exceed $\rho\cdot T_k(i-1)$. 
\end{proof}

Given $f:\{1,\dots,6\}\to \mathbb{N}$ to be fixed later, we define, for every $i\in \{1,\dots,6\}$, the  set of nodes
\[
\bad(i)=\{s\in N_G(u_0) \mid (\mbox{color}(s)=-1) \land (s\notin C_{12} \lor C_{14})\land (\rho(s)>f(i))\},
\]
where $\rho(s)$ denotes the maximum number of node-disjoint well-colored paths from~$s$ to node~$u_i$ in the graph.
Thanks to Lemma~\ref{lem-paths-from-single-s}, we have that a neighbor of $u_0$ colored $-1$ and not in any 12- or 14-cycle, causing $u_i$ to receive more than $T_7(i)$ identifiers, is in $\bad(i)$. This set of nodes thus represents the \emph{bad} neighbors of $u_0$, those that will prevent us from detecting any cycle whenever any such neighbor is picked.

The rest of this section will prove that the bad nodes represent only a fraction of the neighbors of $u_0$. It follows that, by performing sufficiently many choices of~$s$, the probability to select a \emph{good} neighbour $s$ of $u_0$, which will not cause congestion, and will thus allow detecting the cycle, is still $\Omega(n^{-6/7})$.  The parameter $f(i)$  makes the connection between the parameter $T_7(i)$ used by the algorithm, and the set of nodes we do not want to pick as the source~$s$. Formally we are aiming at showing the following result.

\begin{proposition}\label{prop_res_1214}
Let us set $f(1)=60$, $f(2)=f(3)=10$, $f(4)=f(5)=5$, and $f(6)=6$. With this setting, we get 
$\big|\bigcup_{i=1}^{6}\bad(i)\big|\leq \frac{35}{72}\deg(u_0)+3$.
\end{proposition}

The thresholds yielded by the function $f$ defined in Proposition~\ref{prop_res_1214} are: 
\[
\begin{array}{lll}
T_7(1)=60 & \;\;\; T_7(2)=600 & \;\;\; T_7(3)=6\,000 \\
T_7(4)=30\,000 & \;\;\;  T_7(5)=150\,000 & \;\;\; T_7(6)=900\,000
\end{array}
\] 
To prove Proposition~\ref{prop_res_1214}, our strategy is to bound each $|\bad(i)|$ separately by a fraction of the degree of $u_0$ (for $i=1,2,3$), or by a constant (for $i=4,5,6$). Let us now consider the values of $i=1,\dots,6$ successively.

\subsection{Case $i=1$.}

Our goal is to bound $|\bad(1)|$, given  $f(1)=60$. To achieve that, we will show that the nodes colored 0 whose identifiers can reach $u_1$ when a node $s\in \bad(1)$ is picked have to be sufficiently many compared to $\bad(1)$ itself. Otherwise a 12-cycle that involves nodes from $\bad(1)$ would appear, which contradicts the definition of the set of bad nodes. 

\begin{lemma}\label{lem_i1}
If $f(1)\geq 60$ then $|\bad(1)|\leq \frac{1}{4}\deg(u_0)$.
\end{lemma}

\begin{proof}
    Let $W_0$ denote the set of nodes $x\neq u_0$ colored 0  such that $x$ is a heavy neighbor of a node in $\bad(1)$, and a neighbor of $u_1$. This means that for any node $s\in \bad(1)$ that is picked, the identifiers  that $u_1$ receives are those of $u_0$ and of nodes in $W_0$. Let us then consider the bipartite graph $H$ formed by nodes of $\bad(1)$ and $W_0$, and the edges between $\bad(1)$ and $W_0$. Let $H'$ be the subgraph of $H$ obtained by iteratively deleting all nodes of degree at most 11. If $H'$ is not empty, then,  since all its vertices have degree at least 12, we can construct a path of length 11 starting from any vertex of $H'$. Thanks to the fact that $H'$ is bipartite, this path  has either both endpoints in $\bad(1)$, or both in $W_0$, meaning that they are linked to $u_0$ or $u_1$, creating a 12-cycle with the path. This cannot be true as it would mean that some nodes in $\bad(1)$ are in a 12-cycle. 
It follows that $H'$ is empty. As a consequence, 
    $$ 60\cdot |\bad(1)|\leq f(1) \cdot |\bad(1)|\leq |E(H)|< 12\, (|\bad(1)|+|W_0|),$$
where the second inequality comes the fact that any node in $\bad(1)$ has a degree larger than $f(1)$ in $H$, and the third inequality comes from the fact that our iterative removing of nodes of degree at most 11 in $H$ has removed all of the nodes. This yields  $|W_0|>4|\bad(1)|$.
Under our assumption that $u_0$ has maximum degree in $C^\star$, we then get $\deg(u_0)\geq\deg(u_1)\geq|W_0|\geq 4\;|\bad(1)|$.

\end{proof}

\subsection{Case $i=2$.}

To prove upper bounds on the number of bad nodes for $i=2$, as well as for $i> 2$, we use the following lemma that allows us to assume the existence of node-disjoint well-colored paths from different nodes in $\bad(i)$ to $u_i$.

\begin{lemma}\label{lem_paths_from_B}
Let $b\geq 1$, and let $U$ be a set of nodes. If $f(i)\geq (b-1)i+|U|$ then either $|\bad(i)|<b$, or, for any $c\in\{0,\dots,b-1\}$, any nodes $s^1,\dots,s^c\in \bad(i)$, any collection $\mathcal{C}$ of $c$~node-disjoint well-colored paths from $\bad(i)$ to $u_i$ that do not intersect $U$, and any $s\in \bad(i)\smallsetminus\{s^1,\dots,s^c\}$, there exists well-colored path~$P$
 that does not intersect $U$ nor any path in~$\mathcal{C}$.
\end{lemma}

\begin{proof} 
Let $\mathcal{C}=\big \{P^j=(s^j,w_0^j,\dots,w_{i-1}^j) \mid j=1,\dots,c\big \}$. A well-colored path from any node $s\in \bad(i)$ to $u_i$ cannot go through any other node in $\bad(i)$ as the bad nodes are colored $-1$. Such a path may however contain some nodes in $U$, or in the paths in~$\mathcal{C}$.  There are less than $ci+|U|$ such nodes in total. Since $f(i)\geq ci+|U|$, any node $s\in \bad(i)\smallsetminus\{s^1,\dots,s^c\}$ has at least $ci+|U|+1$ node-disjoint well-colored paths to $u_i$. Therefore, there is a well-colored path from $s$ to $u_i$ that does not contain any node in $U$ nor any node in the paths in~$C$, as claimed. 

\end{proof}

To find upper bounds on the sizes of $\bad(2)$ and $\bad(3)$, the strategy is similar to the case $i=1$. The novelty is to consider each color of the node-disjoint paths from nodes in $\bad(i)$ to $u_i$, and to show that the paths cannot \emph{merge} at nodes of the considered color without making the nodes of $\bad(i)$ appear in a 12- or 14-cycle. 

\begin{lemma}\label{lem_i2}
If $f(2)\geq 10$ then $|\bad(2)|\leq \frac{1}{9}\deg(u_0)+2$.
\end{lemma}

\begin{proof}
We say that two well-colored paths from $\bad(i)$ to $u_i$ \emph{merge at color~$j$} if (1)~they are node-disjoint before color $j$, and (2)~they use the same nodes of color~$j$. We first consider paths merging at color~1.

\begin{ourclaim}\label{i2_color1}
If $f(2)\geq 6$, $|\bad(2)|\geq 3$,  and there exist two distinct nodes $s^1,s^2\in \bad(2)$ with paths merging at color~1, then one of those paths goes through $u_0$.
\end{ourclaim}

\begin{claimproof}
For the purpose of contraposition, let us assume that there exist two distinct nodes $s^1,s^2$ in $\bad(2)$ that have well-colored paths $P^1=(s^1,w_0^1,x_1)$, and $P_z=(s^2,z_0,x_1)$ to $u_2$ merging at color 1, with $w_0^1\neq u_0$ and $z_0 \neq u_0$.
Then, by applying Lemma~\ref{lem_paths_from_B} with $b=3$ and $U_1=\{u_0,z_0\}$, we get that, if $|\bad(2)|\geq 3$, then (1)~there exists a well-colored path $P^2=(s^2,w_0^2,w_1^2)$ to $u_2$ that does not intersect $U_1\cup P^1$, and (2)~for all $s^3\in \bad(2)\smallsetminus \{s^1,s^2\}$, there is a well-colored path $P^3=(s^3,w_0^3,w_1^3)$ to $u_2$ that does not intersect $U_1\cup P^1\cup P^2$.
As a consequence,  
$$
(u_0,s^1,w_0^1,x_1,z_0,s^2,w_0^2,w_1^2,u_2,w_1^3,w_0^3,s^3)
$$
is a 12-cycle (see Fig.~\ref{fig_i2}-right), which contradicts $s^1,s^2,s^3\in \bad(2)$. This means that nodes in $\bad(2)$ cannot have paths that merge at color 1 except if one of these paths goes through~$u_0$.
\end{claimproof}

We now consider paths merging at color 0. 

\begin{ourclaim}\label{i2_color0}
Let $s^1,s^2\in \bad(2)$ with $s^1\neq s^2$, and let $x_0\neq u_0$ such that $s^1$ has a well-colored path $P^1=(s^1,x_0,w_1^1)$ to $u_2$, and $s^2\in N(x_0)$. If $f(2)\geq 8$ then there are no two distinct nodes $s^3,s^4\in \bad(2)\smallsetminus\{s^1,s^2\}$ having paths merging at a node colored 0 different from $u_0$ and $x_0$.
\end{ourclaim}

\begin{claimproof}
Let us assume, for the purpose of contradiction, that there exist $s^3,s^4\in \bad(2)\smallsetminus \{s^1,s^2\}$ that have well-colored paths $P^3=(s^3,y_0,w_1^3)$ and $P_y=\{s^4,y_0,w_1^3\}$ to~$u_2$, merging at color~0 with $y_0 \neq u_0, x_0$. 
Applying Lemma~\ref{lem_paths_from_B} with $b=4$ and $U=\{u_0,x_0,y_0\}$, we get that if $|\bad(2)|\geq 4$, then
there is a well-colored path $P^4=(s^4,w_0^4,w_1^4)$ to $u_2$ that does not intersect $U\cup P^1\cup P^2\cup P^3$.
It follows that 
$$
(u_0,s^1,x_0,s^2,w_0^2,w_1^2,u_2,w_1^4,w_0^4,s^4,y_0,s^3)
$$
is a 12-cycle (see Fig.~\ref{fig_i2}-left), which contradicts the fact that $s^1,s^2,s^3,s^4\in \bad(2)$. This means that nodes in $\bad(2)\smallsetminus \{s^1,s^2\}$ cannot have well-colored paths to $u_2$ merging at any node colored 0 different than $u_0$ or $x_0$.
\end{claimproof}

\begin{figure}[!h]
\centering
\includegraphics[scale=0.5]{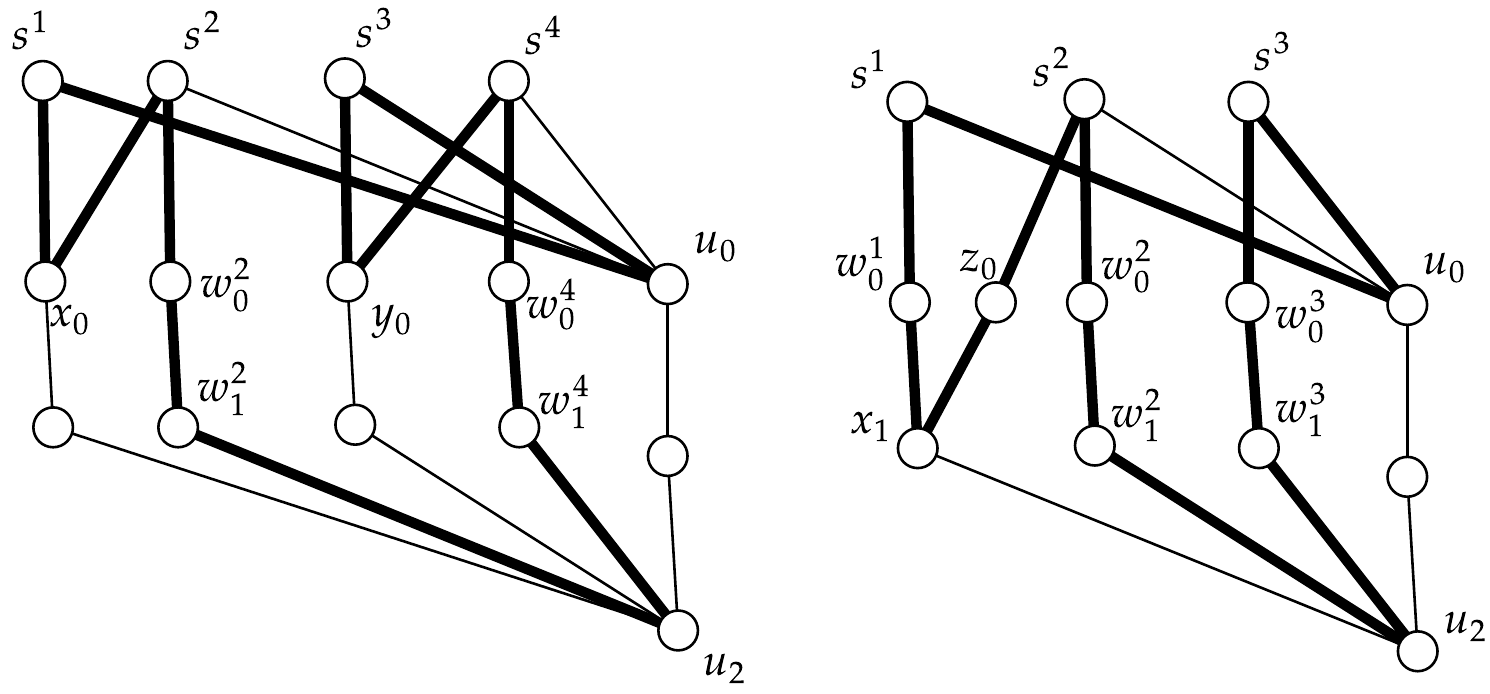}
\caption{Bold cycles are 12-cycles appearing whenever nodes in $\bad(2)$ have merged paths.}
\label{fig_i2}
\end{figure}

In the end, by combining the impossibility results of Claims~\ref{i2_color1} and~\ref{i2_color0}, two situations can occur. The first scenario is that the nodes $s^1$ and $s^2$ defined in Claim~\ref{i2_color0} do not exist, and any two nodes $s,s'\in \bad(2)$ cannot merge their well-colored paths to~$u_2$, except in $u_0$. As every node in $\bad(2)$ has at least $f(2)+1$ node-disjoint well-colored paths to $u_2$. By discarding (if it exists) the one going through $u_0$, we still have $f(2)$ paths not merging with any other well-colored path from $\bad(2)$ to $u_2$. It follows that $f(2)\cdot|\bad(2)|\leq \textrm{deg}(u_2)$.

The other scenario is that the nodes $s^1$ and $s^2$ as in Claim~\ref{i2_color0} do exist. In this case, any  other two nodes  $s,s'\in \bad(2)$ cannot merge their well-colored paths to $u_2$, except in $u_0$ or $x_0$. Discarding paths going through those two nodes, any node $s\in \bad(2)\smallsetminus\{s^1,s^2\}$ still has at least $f(2)-1$ paths not merging with any other well-colored path from $\bad(2)\smallsetminus\{s^1,s^2\}$ to $u_2$. It follows  that 
$$(f(2)-1)\cdot(|\bad(2)|-2)\leq \textrm{deg}(u_2).$$
Therefore, in all cases, we have 
$(|\bad(2)|-2)(f(2)-1) \leq \textrm{deg}(u_2)\leq \deg(u_0),$
which proves Lemma~\ref{lem_i2}.
\end{proof}

\subsection{Case $i=3$.}

We show that the nodes in $\bad(3)$ cannot have their paths merging before reaching $u_3$ (see proof in Appendix~\ref{app:lem_i3}). 

\begin{lemma}\label{lem_i3}
If $f(3)\geq 10$ then $ |\bad(3)|\leq \frac{1}{8}\deg(u_0)$.
\end{lemma}

\subsection{Cases $i\in \{4,5,6\}$.}

For $i=4,5,6$, with values of $f(i)$ satisfying the inequality in the statement of Lemma~\ref{lem_paths_from_B}, the mere existence of one or two nodes in $\bad(i)$ is impossible, as such nodes would appear in a 12- or 14-cycle. This is shown below (see proof in Appendix~\ref{app:lem_i456}). 

\begin{lemma}\label{lem_i456}
The following holds: 
\begin{itemize}
\item If $f(4)\geq 5$ then $|\bad(4)| \leq 1$.
\item If $f(5)\geq 5$ then $\bad(5)=\varnothing$.
\item If $f(6)\geq 6$ then $\bad(6)=\varnothing$.
\end{itemize}
\end{lemma}

Proposition~\ref{prop_res_1214} directly follows from Lemmas~\ref{lem_i1}, and \ref{lem_i2}-\ref{lem_i456}. 
Since, for every  $i\in\{1,\dots,6\}$, $T_7(14-i)=T_7(i)$ induces the same upper bound for $|\bad(14-i)|$ as for $|\bad(i)|$, we get that $\big |\bigcup_{i\in\{1,\dots,13\}\smallsetminus \{7\}} \bad(i)\big|\leq \frac{35}{36}\deg(u_0)+6$.
It follows that  
$
\big |N(u_0)\smallsetminus \bigcup_{i\in\{1,\dots,13\}\smallsetminus \{7\}} \bad(i)\big|\geq \frac{1}{36}\deg(u_0)-6.
$
As a consequence, the number of \emph{good} neighbours $s$ of $u_0$ (not belonging to any $B(i)$ is at least a constant fraction of deg($u_0$). This means that after $\Theta(n^{6/7})$ repetitions of the choice of~$s$,  a node in $N(u_0)\smallsetminus \bigcup_{i\in\{1,\dots,13\}\smallsetminus \{7\}} \bad(i)$ that is colored~$-1$ will be picked with probability at least 2/3. By the previous Lemmas, this node will lead to a rejection, detecting a 12- or 14-cycle.

\subsection{Looking for a $C_{12}$.}

At this point, we can assume that there are no 14-cycle in the graph because, if there were, then the algorithm would have rejected before. Let us then assume the existence of a 12-cycle $C^\star=\{u_0,\dots,u_{11}\}$, that is well colored, where $u_0$ a heavy node. Then, by fixing
$T_6(i) = 1$
for every $i=1,\dots,5$, 
$u_6$ will reject.
Indeed, recall that there are $\Omega(n^{1/6})$ neighbors of $u_0$ colored~$-1$. Therefore, by performing $O(n^{6/7})$ iterations, the probability of picking $s\in N(u_0)$ colored~$-1$ is at least~$2/3$. Whenever such a node~$s$ is picked, $u_0$~sends its identifier. Suppose that $u_i$ is the first node in $C^\star$ to receive 2 identifiers. Then, by Lemma~\ref{lem-paths-from-single-s}, one of these identifiers comes from a path $w_0,w_1,...,w_{i-1}$ that is node-disjoint from $\{u_0,\dots,u_{i-1}\}$. Therefore, 
$ (s,w_0,w_1,...,w_{i-1},u_i,u_{i+1},...,u_{11},u_0)$
is a 14-cycle, a contradiction, which completes the proof of Theorem~\ref{theo:small-families}.

\paragraph{\textit{Remark.}}

A threshold algorithm deciding $\{C_{10},C_{12}\}$-freeness in $O(n^{1-1/6})$ rounds is given in Appendix~\ref{app:C10C12}.




\section{Conclusion} 

The threshold-based approach, as used in~\cite{Censor-HillelFG20}, is appealing  for the design of efficient \CONGEST\/ algorithms deciding $C_{2k}$-freeness, for arbitrary $k\geq 2$. It was successfully applied to $k\in\{2,\dots,5\}$, resulting in algorithms deciding $C_{2k}$-freeness in $O(n^{1-1/k})$ rounds. We have shown that it is hopeless to use the threshold-based approach as such for $k\geq 6$. 

Nevertheless, we have also shown that, despite this limit, the threshold-based approach can be used to design algorithms for deciding $\{C_{12},C_{14}\}$-freeness in $n^{1-\frac{1}{7}}$ rounds, even if neither $C_{12}$-freeness nor $C_{14}$-freeness can be decided in the same round-complexity by threshold-based algorithms. We do not know whether this is just a specific case, or whether there is an infinite collection of pairs $(k,k')$ with $k>k'$ such that $\{C_{2k'},C_{2k}\}$-freeness can be decided in $O(n^{1-1/k})$ rounds by a threshold-based algorithm. 

So far, the best known generic algorithm, i.e., an algorithm applying to all $k\geq 2$, decides $C_{2k}$-freeness in  $n^{1-1/\Theta(k^2)}$ rounds~\cite{fischer2018}, and it is open whether one can do better in general. It may actually be the case that the round-complexity of deciding $C_{2k}$-freeness is precisely $\Theta(n^{1-1/k})$ for all $k\geq 2$, but this is just a guess. 

\paragraph{Acknowledgements:} The authors are thankful to Pedro Montealegre and Ivan Rapaport for fruitful initial discussions about the power and limit of the threshold-based algorithms from~\cite{Censor-HillelFG20}.

\newpage
\bibliographystyle{splncs04}
\bibliography{biblio-cycle}

\begin{thebibliography}{10}
\providecommand{\url}[1]{\texttt{#1}}
\providecommand{\urlprefix}{URL }
\providecommand{\doi}[1]{https://doi.org/#1}

\bibitem{AlonYZ95}
Alon, N., Yuster, R., Zwick, U.: Color-coding. J. {ACM}  \textbf{42}(4),
  844--856 (1995)

\bibitem{BroersmaGPS12}
Broersma, H., Golovach, P.A., Paulusma, D., Song, J.: Determining the chromatic
  number of triangle-free {2P3}-free graphs in polynomial time. Theor. Comput.
  Sci.  \textbf{423},  1--10 (2012)

\bibitem{censor2022}
Censor-Hillel, K.: Distributed subgraph finding: Progress and challenges. arXiv
  preprint arXiv:2203.06597  (2022)

\bibitem{Censor-HillelFG20}
Censor{-}Hillel, K., Fischer, O., Gonen, T., Gall, F.L., Leitersdorf, D.,
  Oshman, R.: Fast distributed algorithms for girth, cycles and small
  subgraphs. In: 34th International Symposium on Distributed Computing (DISC).
  LIPIcs, vol.~179, pp. 33:1--33:17. Schloss Dagstuhl - Leibniz-Zentrum
  f{\"{u}}r Informatik (2020)

\bibitem{censor2015}
Censor-Hillel, K., Kaski, P., Korhonen, J.H., Lenzen, C., Paz, A., Suomela, J.:
  Algebraic methods in the congested clique. In: 34th ACM Symposium on
  Principles of Distributed Computing (PODC). pp. 143--152 (2015)

\bibitem{drucker2014}
Drucker, A., Kuhn, F., Oshman, R.: On the power of the congested clique model.
  In: 33rd ACM symposium on Principles of distributed computing (PODC). pp.
  367--376 (2014)

\bibitem{eden2022}
Eden, T., Fiat, N., Fischer, O., Kuhn, F., Oshman, R.: Sublinear-time
  distributed algorithms for detecting small cliques and even cycles.
  Distributed Comput.  \textbf{35}(3),  207--234 (2022)

\bibitem{EvenFFGLMMOORT17}
Even, G., Fischer, O., Fraigniaud, P., Gonen, T., Levi, R., Medina, M.,
  Montealegre, P., Olivetti, D., Oshman, R., Rapaport, I., Todinca, I.: Three
  notes on distributed property testing. In: 31st International Symposium on
  Distributed Computing (DISC). LIPIcs, vol.~91, pp. 15:1--15:30. Schloss
  Dagstuhl - Leibniz-Zentrum f{\"{u}}r Informatik (2017)

\bibitem{fischer2018}
Fischer, O., Gonen, T., Kuhn, F., Oshman, R.: Possibilities and impossibilities
  for distributed subgraph detection. In: 30th on Symposium on Parallelism in
  Algorithms and Architectures (SPAA). pp. 153--162 (2018)

\bibitem{FraigniaudO19}
Fraigniaud, P., Olivetti, D.: Distributed detection of cycles. {ACM} Trans.
  Parallel Comput.  \textbf{6}(3),  12:1--12:20 (2019)

\bibitem{korhonen2017}
Korhonen, J.H., Rybicki, J.: Deterministic subgraph detection in broadcast
  congest. arXiv preprint arXiv:1705.10195  (2017)

\bibitem{Kovari1954}
K{\H{o}}v{\'a}ri, P., T~S{\'o}s, V., Tur{\'a}n, P.: On a problem of
  zarankiewicz. In: Colloquium Mathematicum. vol.~3, pp. 50--57. Polska
  Akademia Nauk (1954)

\bibitem{Neuen2022}
Neuen, D.: Isomorphism testing for graphs excluding small topological
  subgraphs. In: 33rd ACM Symposium on Discrete Algorithms (SODA). pp.
  1411--1434 (2022)

\bibitem{Peleg2000}
Peleg, D.: Distributed computing: a locality-sensitive approach. SIAM (2000)

\end{thebibliography}


\newpage 
\appendix
\centerline{\Large\bf A P P E N D I X}

\section{Proof of Lemma~\ref{lem:C*unique}}
\label{app:lem:C*unique}

We can view $G_k$ as a weighted graph $\widehat{G}_k$, with the following edge-weights (see Fig.~\ref{Gasaweightedgraph}): 
\begin{itemize}
\item For every $p\in\{1,\dots,N\}$, the edge between $u_0$ and $s^p$ has weight~1; Similarly, for every $q\in\{1,\dots,N\}$, the edge between $u_{k-3}$ and $w_{k-4}^q$ has weight~1;
\item All the part of $G_k$ including the nodes $w_j^{p,q}$ for $(p,q)\in\{1,\dots,N\}^2$ and $j\in\{0,\dots,k-5\}$ is replaced by a complete bipartite graph $K_{N,N}$ with partitions $S$ and~$W$, with edge-weights~$k-3$;
\item The path $u_0,\dots,u_{k-3}$ is replaced by two parallel edges $e_1$ and $e_2$ with respective weights~$k-3$ and~$k+3$;
\item The private neighbors of $w_0^{p,q}$ and $w_{k-5}^{p,q}$, $(p,q)\in\{1,\dots,\alpha\}^2$, are simply discarded. 
\end{itemize}
Note that there is a one-to-one correspondence between the cycles in~$G_k$ and the cycles in~$\widehat{G}_k$. Moreover, if the lengths of the cycles in $\widehat{G}_k$ take into account the edge-weights, then this correspondence also preserve the lengths. 

\begin{figure}[!ht]
\centering
\makebox[\textwidth][c]{\includegraphics[scale=0.6]{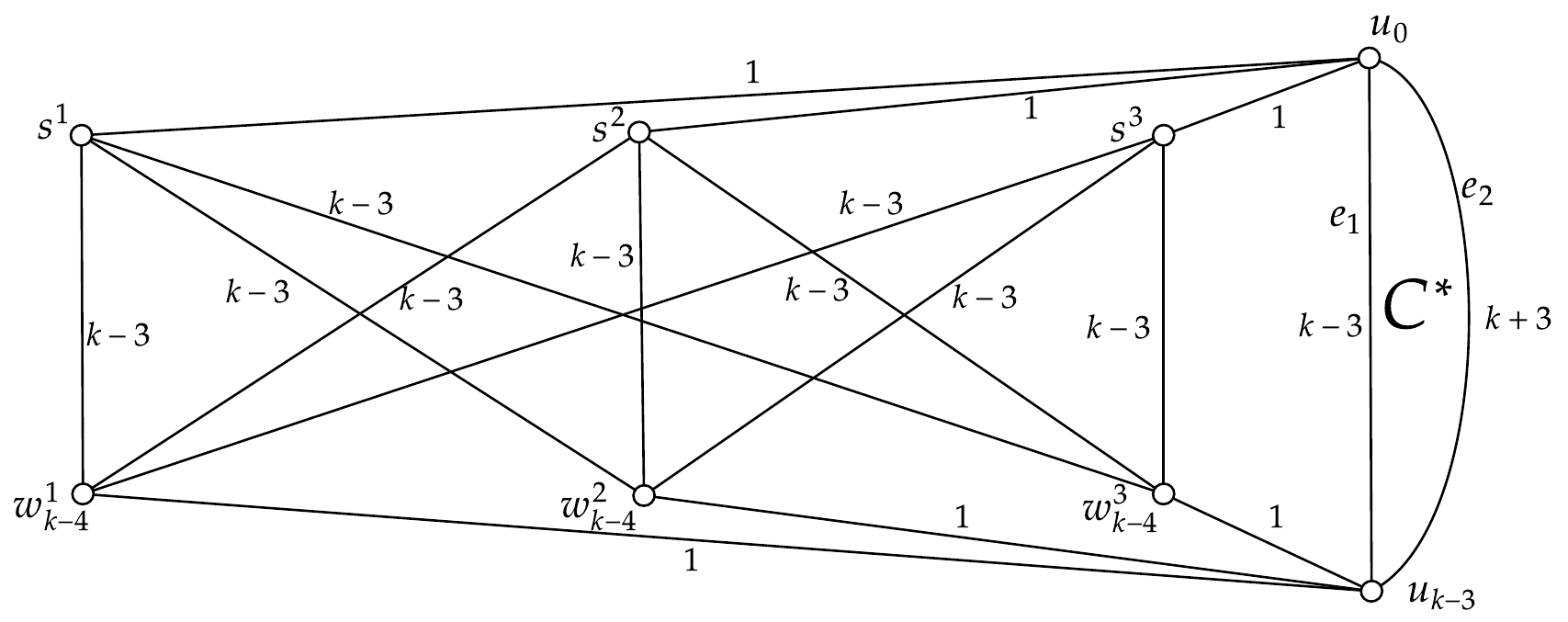}}
\caption{The weighted graph $\widehat{G}_k$ for $k\geq 7$ and $N=3$.}
\label{Gasaweightedgraph}
\end{figure}

Any cycle in $\widehat{G}_k$ different from $C^\star=(e_1,e_2)$, but using $e_2$ has length larger than~$2k$. Any cycle in $\widehat{G}_k$ different from $C^\star=(e_1,e_2)$, but using $e_1$ must contain an odd number $2x+1$ of edges from the complete bipartite subgraph $K_{N,N}$. Therefore it has length $k-3+1+(2x +1)(k-3)+1$ for some integer $x \geq 0$. For $x=0$, this length is~$2k-4$, and, for $x\geq 1$, this length is at least  $4k-10>2k$ for every $k\geq 6$. 
For similar reasons, every cycle passing through $u_0$ or $u_{k-3}$, but not both, is also of length $k-3+1+(2x +1)(k-3)+1$ for some integer $x \geq 0$, and the same analysis holds.
Every cycle passing through $u_0$ and $u_{k-3}$ but not using~$e_1$ nor~$e_2$ contains an even number of edges from the $K_{N,N}$. Thus is has a length of the form $2(k-3)+4+2x(k-3)$ for some integer $x\geq 0$. For $x=0$, this length is~$2k-2$, and, for $x\geq 1$, this length is at least $4k-8>2k$ for any $k\geq 6$.
Finally, every cycle fully included in the complete bipartite graph~$K_{N,N}$ has length of the form $2(x+2)(k-3)$ for some integer $x\geq 0$. This length is at least $4k-12>2k$, for every $k\geq 7$.
\qed

\section{The Specific Construction for $k=6$}
\label{app:specific-for-k6}

Lemma~\ref{lem:C*unique} does not hold for $k=6$. Indeed, any 4-tuple $(s,s',w_2,w_2')\in S^2\times W^2$ induces a 12-cycle in the complete bipartite graph $K_{N,N}$ in $\widehat{G}_6$. Nevertheless, there are no $C_{12}$ in $G_6$ passing through $u_0$ or $u_3$, other than~$C^\star$. We slightly modify~$G_6$ to extend Lemma~\ref{lem:C*unique} to $k=6$. Specifically, we replace the complete bipartite graph  in $\widehat{G_6}$  by a dense bipartite graph with no 4-cycle, using the following lemma.

\begin{lemma}[\cite{Kovari1954}]
There exists an infinite family of $C_4$-free graphs $\{G_d\mid d \;\mbox{prime}\}$ such that, for every prime number~$d$, $G_d$~is a $d$-regular bipartite graph in which each partition has size $d^2$, and $|E(G_d)|=d^3+o(d^3)$. 
\end{lemma}

Let $N=d^2$. In $\widehat{G}_6$, we replace the complete bipartite graph $K_{N,N}$ by $G_d$. With this modification, Lemma~\ref{lem:C*unique} holds. We now revisit Lemma~\ref{lem:lowsuccessfork7}. $X$~becomes a random variable following a Binomial law with parameters $\sqrt{N}$ and $(\frac{1}{12})^{3}$. As a consequence, 
\[
\Pr[X\leq T] = O(N^T e^{-\Theta(\sqrt{N})}). 
\]
Therefore, the probability that $u_3$ has to forward at most $T$ messages is $O(N^Te^{-\Theta(\sqrt{N})})$. If $T=o(\sqrt{N}/\log N)\simeq o(n^{1/6}/\log n)$, then this probability is asymptotically equal to $\exp(-\Theta(n^{1/6}))$. So for $k=6$ we obtain a result similar to Lemma~\ref{lem:lowsuccessfork7}, by simply replacing the value of threshold $T$ by $o(\sqrt{N}/\log N)\simeq o(n^{1/6}/\log n)$. The proof of Theorem~\ref{theo:does-not-extend} follows as for $k\geq 7$. 
\qed

\section{Proof of Lemma~\ref{lem_i3}}
\label{app:lem_i3}

We first consider paths merging at color 2.

\begin{ourclaim}\label{i3_color2}
    If $f(3)\geq 8$, and if there exist two distinct nodes $s^1,s^2\in \bad(3)$ that have well-colored paths merging at a node of color 2 different from $u_2$, then one of these paths goes through $u_0$ or~$u_1$.
\end{ourclaim}

\begin{claimproof}
Suppose that there exist $s^1,s^2\in \bad(3)$ that have well-colored paths $P^1=(s^1,w_0^1,w_1^1,x_2)$ and $P_z=(s^2,z_0,z_1,x_2)$ to $u_3$ merging at color 2 with $u_0\notin\{z_0, w_0^1\}$, $u_1\notin\{z_1,w_1^1\}$, and $x_2\neq u_2$.
Then, applying Lemma~\ref{lem_paths_from_B} with $b=2$ and $U=\{u_0,u_1,u_2,z_0,z_1\}$, we get that if $|\bad(3)|\geq 2$, then there is a well-colored path $P^2=(s^2,w_0^2,w_1^2,w_2^2)$ to $u_3$ that does not intersect $U\cup P^1$.
Therefore 
$$ (u_0,s^1,w_0^1,w_1^1,x_2,z_1,z_0,s^2,w_0^2,w_1^2,w_2^2,u_3,u_2,u_1)$$ 
is a 14-cycle (see Fig.~\ref{fig_i3}-right), which contradicts $s^1,s^2\in \bad(3)$. This means that nodes in $\bad(3)$ cannot have well-colored paths to $u_3$ merging at any node colored 2 (except if they go through $u_0$ or $u_1$, or merge at $u_2$).
\end{claimproof}

We now consider paths merging at color 1.

\begin{ourclaim}\label{i3_color1}
    If $f(3)\geq 7$, and if there exist two distinct nodes $s^1,s^2\in \bad(2)$ that have well-colored paths merging at a node of color~1 different from $u_1$, then one of these paths goes through~$u_0$, or they both go through~$u_2$.
\end{ourclaim}

\begin{claimproof}
Suppose that there exist $s^1,s^2\in \bad(3)$ that have well-colored paths $P^1=(s^1,w_0^1,x_1,w_2^1)$ and $P_z=(s^2,z_0,x_1)$ to $u_3$ merging at color~1, with $u_0\notin\{w_0^1,z_0\}$ and $w_2^1\neq u_2$. 
Applying Lemma~\ref{lem_paths_from_B} with $b=2$ and $U=\{u_0,u_1,u_2,z_0\}$, we get that  if $|\bad(3)|\geq 2$,  then there is a well-colored path $P^2=(s^2,w_0^2,w_1^2,w_2^2)$ to $u_3$ that does not intersect $U\cup P^1$.
This implies that 
$$ (u_0,s^1,w_0^1,x_1,z_0,s^2,w_0^2,w_1^2,w_2^2,u_3,u_2,u_1)$$ 
is a 12-cycle (see Fig.~\ref{fig_i3}-center), which contradicts $s^1,s^2\in \bad(3)$. This means that nodes in $\bad(3)$ cannot have well-colored paths to $u_3$ merging at any node colored~1 other than $u_1$, except if it goes through $u_0$ or~$u_2$.
\end{claimproof}

We finally consider paths merging at color 0.
\begin{ourclaim}\label{i3_color0}
    If $f(3)\geq 8$ and $|\bad(3)|\geq 3$, then there are no two distinct nodes $s^1,s^2\in \bad(2)$ that have paths merging at a node of color~0 different from $u_0$.
\end{ourclaim}

\begin{claimproof}
Suppose that there exists $s^1,s^2\in \bad(3)$ that have well-colored paths $P^1=(s^1,x_0,w_1^1,w_2^1)$ and $P_x = \{s^2,x_0,w_1^1,w_2^1\}$ to $u_3$ merging at color~0, with $x_0\neq u_0$.
Applying Lemma~\ref{lem_paths_from_B} with $b=3$ and $U=\{u_0,x_0\}$, we get that if $|\bad(3)|\geq 3$, then (1)~there is a well-colored path $P^2=(s^2,w_0^2,w_1^2,w_2^2)$ to $u_3$ that does not intersect $U\cup P^1$, and (2)~for all $s^3\in \bad(3)\smallsetminus \{s^1,s^2\}$, there is a well-colored path $P^3=(s^3,w_0^3,w_1^3,w_2^3)$ to $u_3$ that does not intersect $U\cup P^1\cup P^2$.
Therefore, $$ (u_0,s^1,x_0,s^2,w_0^2,w_1^2,w_2^2,u_3,w_2^3,w_1^3,w_0^3,s^3)$$
is a 12-cycle (see Fig.~\ref{fig_i3}-left), which contradicts $s^1,s^2,s^3\in \bad(3)$. This means that nodes in $\bad(3)$ cannot have well-colored paths to $u_3$ merging at color~0 (except $u_0$).
\end{claimproof}

\begin{figure}[!ht]
\centering
\includegraphics[scale=0.6]{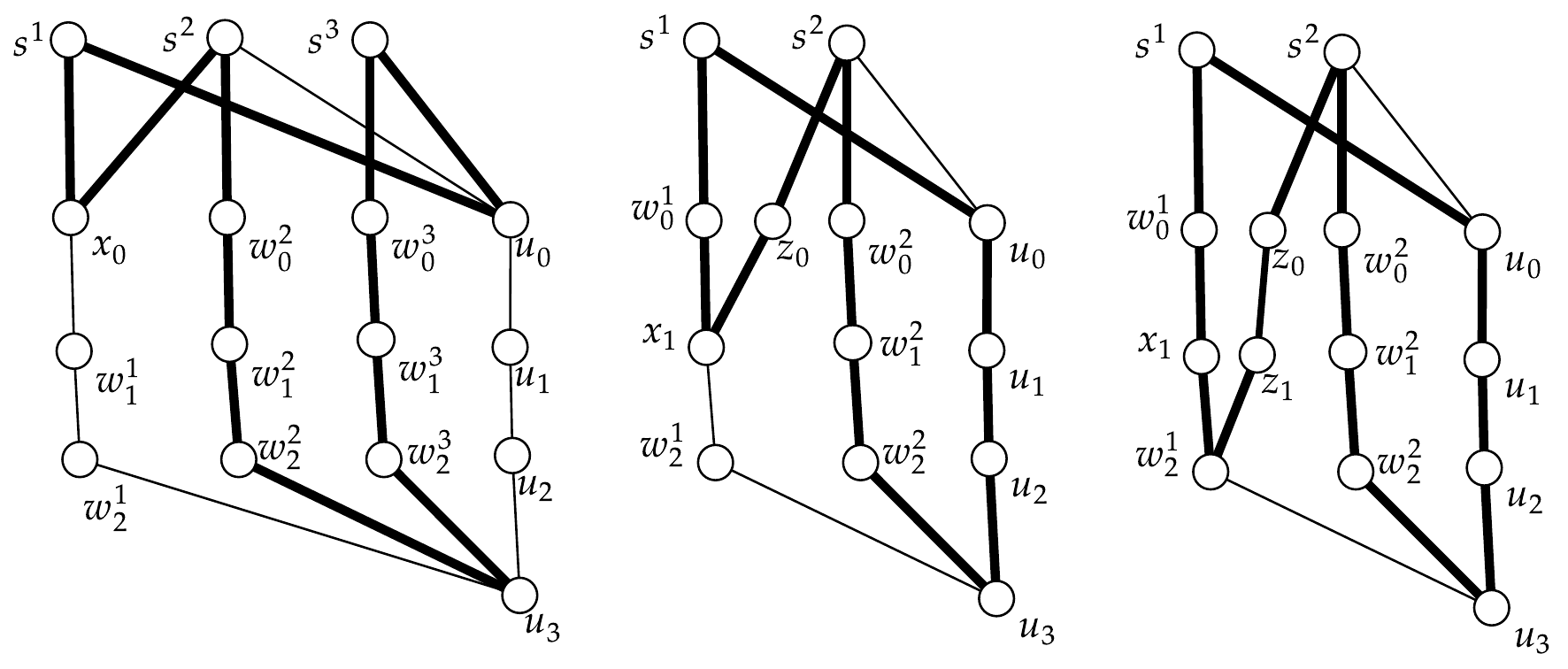}
\caption{Bold cycles are 12- and 14-cycles appearing whenever nodes in $\bad(3)$ have merged paths.}
\label{fig_i3}
\end{figure}

By combining the impossibility results of Claims~\ref{i3_color2},~\ref{i3_color1} and~\ref{i3_color0}, it follows that paths from $\bad(3)$ to $u_3$ can only merge if they go through $u_0$, $u_1$ or~$u_2$. Let us discard paths passing through those nodes. As each node in $\bad(3)$ has at least $f(3)+1$ node-disjoint paths to~$u_3$, at least $f(3)-2$ of these paths do not intersect any other path from~$\bad(3)$. Therefore, $(f(3)-2)\cdot |\bad(3)|\leq \textrm{deg}(u_3)\leq \textrm{deg}(u_0)$, which proves Lemma~\ref{lem_i3}.
\qed

\section{Proof of Lemma~\ref{lem_i456}}
\label{app:lem_i456}

We treat each $i=4,\dots,6$ sequentially. 

\begin{ourclaim}\label{lem_i4}
If $f(4)\geq 5$ then $|\bad(4)| \leq 1$.
\end{ourclaim}

\begin{claimproof}
Let $U=\{u_0\}$. According to Lemma~\ref{lem_paths_from_B} applied with $b=2$, we have that if $|\bad(4)|\geq 2$ then (1)~for all $s^1\in \bad(4)$, there is a well-colored path $P^1=(s^1,w_0^1,w_1^1,w_2^1,w_3^1)$ to $u_4$ that does not intersect~$U$, and (2)~for all $s^2\in \bad(3)\smallsetminus \{s^1\}$, there is a well-colored path $P^2=(s^2,w_0^2,w_1^2,w_2^2,w_3^2)$ to $u_4$ that does not intersect $U\cup P^1$.
It follows that $$ (u_0,s^1,w_0^1,w_1^1,w_2^1,w_3^1,u_4,w_3^2,w_2^2,w_1^2,w_0^2,s^2)$$ 
is a 12-cycle (see Fig.~\ref{fig_i456}-left), which contradicts $s^1,s^2\in \bad(4)$. 
\end{claimproof}

The arguments for $i=5,6$ are a reformulation of Observation~2 in~\cite{Censor-HillelFG20}, we give them for the sake of completeness.

\begin{ourclaim}\label{lem_i5}
If $f(5)\geq 5$ then $\bad(5)=\emptyset$.
\end{ourclaim}

\begin{claimproof}
Let $U=\{u_0,u_1,u_2,u_3,u_4\}$. By Lemma~\ref{lem_paths_from_B} applied with $b=1$, if $\bad(5)$ is not empty, then for every $s^1\in \bad(5)$, there is a well-colored path $P^1=(s^1,w_0^1,w_1^1,w_2^1,w_3^1,w_4^1)$ to $u_5$ that does not intersect $U$.
Therefore, $$ (u_0,s^1,w_0^1,w_1^1,w_2^1,w_3^1,w_4^1,u_5,u_4,u_3,u_2,u_1)$$ 
is a 12-cycle (see Fig.~\ref{fig_i456}-center), which contradicts $s^1\in \bad(5)$. 
\end{claimproof}

\begin{figure}
\centering
\includegraphics[scale=0.6]{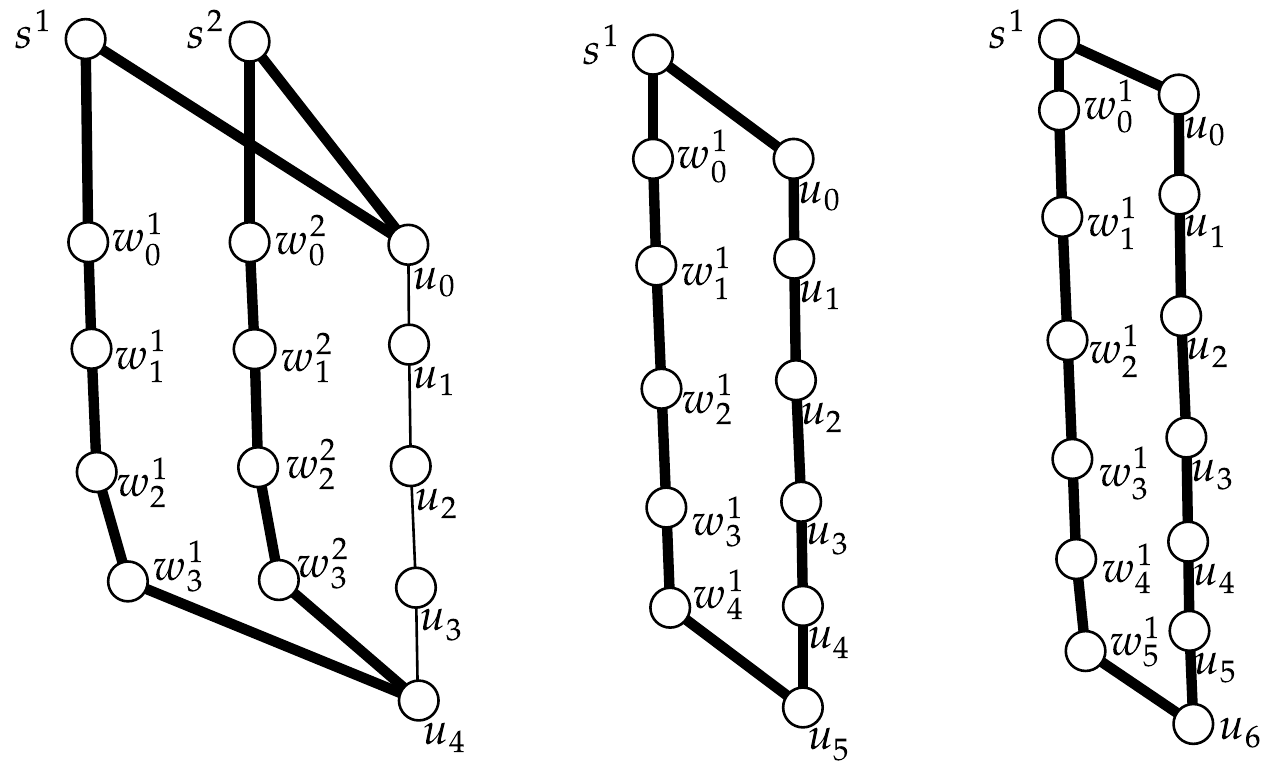}
\caption{Bold cycles are 12- and 14-cycles appearing whenever there are two nodes in $\bad(4)$, or one node in $\bad(5)$, $\bad(6)$.}
\label{fig_i456}
\end{figure}

\begin{ourclaim}\label{lem_i6}
If $f(6)\geq 6$ then $\bad(6)=\emptyset$.
\end{ourclaim}

\begin{claimproof}
Let $U=\{u_0,u_1,u_2,u_3,u_4,u_5\}$. Applying Lemma~\ref{lem_paths_from_B} with $b=1$, we get that if $\bad(6)$ is not empty, then, for every  $s^1\in \bad(6)$, there is a well-colored path $P^1=(s^1,w_0^1,w_1^1,w_2^1,w_3^1,w_4^1,w_5^1)$ to $u_6$ that does not intersect~$U$.
As a consequence, $$ (u_0,s^1,w_0^1,w_1^1,w_2^1,w_3^1,w_4^1,w_5^1,u_6,u_5,u_4,u_3,u_2,u_1)$$ 
is a 14-cycle (see Fig.~\ref{fig_i456}-right), which contradicts $s^1\in \bad(6)$.
\end{claimproof}

The lemma directly follows from Claims~\ref{lem_i4}-\ref{lem_i6}. 
\qed

\section{Deciding $\{C_{10},C_{12}\}$-Freeness}
\label{app:C10C12}

We begin by deciding $C_{10}$-freeness on its own. This is doable using the threshold approach, with the threshold $T_5(i)$, $i=1,\dots,4$, given in~\cite{Censor-HillelFG20}. If a heavy 10-cycle $C=(u_0,\dots,u_9)$ exists, then, by repeating $O(n^{4/5})$ times the choice of~$s$, the probability that $C$ is detected is at least~9/10.
If no node has rejected during the search for a 10-cycle, then we can assume that the graph is $C_{10}$-free. The search for a 12-cycle is still performed by the threshold algorithm. Note that we have proved that deciding $C_{12}$-freeness cannot be done by a threshold algorithm. However, we are now working under the assumption that the graph is $C_{10}$-free. 
Searching for light 12-cycles can be done in $O(n^{5/6})$ rounds. Let us assume that  there exists a heavy 12-cycle $C^\star = \{u_0,\dots, u_{11}\}$ where $u_0$ has maximum degree in the cycle, and, for every~$i=0,\dots,11$, node~$u_i$ has picked color~$i$. Let us fix $T_6(i) = T_5(i)$ for $i=1,\dots,4$. Then there is at least a constant fraction of neighbors $s$ of $u_0$ such that, if $s$ is chosen by the algorithm, then $u_i$ and $u_{12-i}$ receive at most $T_6(i)$ identifiers, for all $1\leq i\leq 4$.

Finally let $T_6(5) = T_6(4)$. This is sufficient because, by Lemma~\ref{lem-paths-from-single-s}, if $u_5$ receives more than $T_6(4)$ identifiers, then there exist two node-disjoint well-colored paths from $s$ to~$u_5$. The combination of these two paths is a 12-cycle. Therefore a node $s$ picked by the algorithm is either in a 12-cycle of his own or will not cause $u_5$ to receive more than $T_6(4)$ identifiers.

\section{Proof of Theorem~\ref{theo:one-every-four}}
\label{sec_incremental}
\label{app:theo:one-every-four}

This section is dedicated to the proof of Theorem~\ref{theo:one-every-four}. In essence, the algorithm deciding $\{C_{4\ell}\mid 1\leq \ell\leq k\}$-freeness performs by successively deciding $C_{4\ell}$-freeness, for $\ell=1,\dots,k$, using threshold-based algorithms. This sequence of algorithms runs in $O(n^{1-1/2k})$ rounds. The same holds for deciding $\{C_{4\ell+2}\mid 1\leq \ell\leq k\}$-freeness, with round-complexity $O(n^{1-1/(2k+1)})$. The proof mostly consists in showing that thresholds can be defined with the guarantee that, if the graph contains a $4\ell$-cycle (or a $(4\ell+2)$-cycle), then this cycle is detected with constant probability. Let us start with $\mathcal{F}_k=\{C_{4\ell}\mid 1\leq \ell\leq k\}$, $k\geq 1$. 

The base case is $\ell=1$, i.e., deciding $C_4$-freeness. It was shown in~\cite{Censor-HillelFG20} that, in this case, a threshold $T_2(1)=1$ suffices.  Let us now assume that we have set up appropriate thresholds for deciding $\mathcal{F}_{k-1}$-freeness. For deciding $\mathcal{F}_k$-freeness, it is sufficient to decide $\mathcal{F}_{k-1}$-freeness first, and then deciding $C_{4k}$-freeness. 
Therefore it is sufficient that the algorithm deciding $C_{4k}$-freeness succeeds whenever the graph is $\mathcal{F}_{k-1}$-free. 
So, from this point on, we assume that the graph has no $4\ell$-cycles, for all $\ell=1,\dots,k-1$. To decide $C_{4k}$-freeness under this latter hypothesis, we set the thresholds as $T_{2k}(0)=1$, and, for $i\geq 1$, 
\[
T_{2k}(i) = \left\{\begin{array}{ll}
        T_{2k}(i-1) & \mbox{if $i$ is odd,}\\
        (i+1) \cdot T_{2k}(i-1) & \mbox{if $i$ is even.}
    \end{array}\right.
\]

To prove that the thresholds work, let $C=(u_0,\dots,u_{4k-1})$  be a $4k$-cycle in the graph, where $u_0$ is the heavy node with largest degree in~$C$. Let us assume that, for every $i=0,\dots,4k-1$, node~$u_i$ is colored~$i$. Note that this occurs with probability at least $1/(4k+1)^{4k}$. Let $s$ be a neighbor of $u_0$ not belonging to any $4k$-cycle, and let us assume that $s$ is colored~$-1$. We define a well-colored path from $s$ to a node $u_i$ as a path of the form $s,w_0,\dots,w_{i-1},u_i$ where, for every $j=0,\dots,i-1$, node $w_j$ is colored~$j$.

To prove that our thresholds are sufficient, we now consider separately the odd and even indices. 

\paragraph{Odd indices.} 

Let $i$ be an odd index with $1\leq i\leq 2k-1$, and let $\rho$ be the maximum number of node-disjoint, well-colored paths from $s$ to~$u_i$. Let us prove by contradiction that  $\rho=1$ (see Fig.~\ref{fig_incremental}(left) for an illustration of the case $\rho \geq  2$).
Indeed, $\rho\geq 1$ as $\{u_0,\dots,u_{i-1}\}$ is a well-colored path from $s$ to~$u_i$.
Suppose now that there exist two node-disjoint well-colored paths from $s$ to $u_i$, denoted by $P$ and~$P'$. Then $\{s\}\cup P\cup \{u_i\}\cup P'$ is a $(2i+2)$-cycle. As $i$ is odd, we get that $2i+2$ is a multiple of 4. If $i\leq 2k-3$, this contradicts our assumption that the graph is $\mathcal{F}_{k-1}$-free. If $i=2k-1$, then $s$ is in a $4k$-cycle that the algorithm finds when it performs color-BFS($s$).
Since $\rho=1$, by Lemma~\ref{lem-paths-from-single-s}, node~$u_i$ receives at most $T_{2k}(i-1)$ identifiers.

\paragraph{Even indices.} 

Let $i$ be an even index with $1\leq i \leq 2k-2$, and let $\rho$ be the maximum number of node-disjoint well-colored paths from $s$ to~$u_i$. 
Let us assume that $s$ is such that $u_i$ receives more than $(i+1) \cdot T_{2k}(i-1)$ identifiers from nodes colored~$i-1$. By Lemma~\ref{lem-paths-from-single-s}, $\rho\geq i+2$. Therefore, there exists a well-colored path from $s$ to $u_{i}$ that does not go through $u_0$. Let us denote this path $P=\{w_0,\dots,w_{i-1}\}$.
Let $s'$ be the next neighbor of $u_0$ colored $-1$ that the algorithm picks. Let us define $\rho'$ as the maximum number of node-disjoint well-colored paths from $s'$ to $u_i$.  If for the source node~$s'$, node $u_i$ also receives more than $(i+1) \cdot T_{2k}(i-1)$ identifiers from nodes colored $i-1$, then, thanks to Lemma~\ref{lem-paths-from-single-s}, $\rho'\geq i+2$. It follows that there exists a well-colored path from $s'$ to $u_i$ that does not go through any of the $i+1$ nodes $u_0, w_0,\dots, w_{i-1}$. Let us denote this path by~$P'$.
As a consequence, $\{s,u_0,s'\}\cup P'\cup \{u_i\} \cup P$ is a $(2i+4)$-cycle (see Fig.~\ref{fig_incremental}(right)). Since $i$ is even, $2i+4$ is a multiple of 4. If $i\leq 2k-4$, this contradicts the assumption that the graph is $\mathcal{F}_{k-1}$-free. If $i=2k-2$, then $s$ is in a $4k$-cycle, which is detected by~$s$.
Therefore, if the algorithm picks a neighbor $s$ of $u_0$ colored $-1$ that causes $u_i$ to receive more than $T_{2k}(i)$ identifiers, then no other neighbor $s'$ of $u_0$ can also cause $u_i$ to receive more than $T_{2k}(i)$ identifiers, as $s$ would then be in a cycle that would be detected by~$s$.

\begin{figure}
\centering
\includegraphics[scale=0.65]{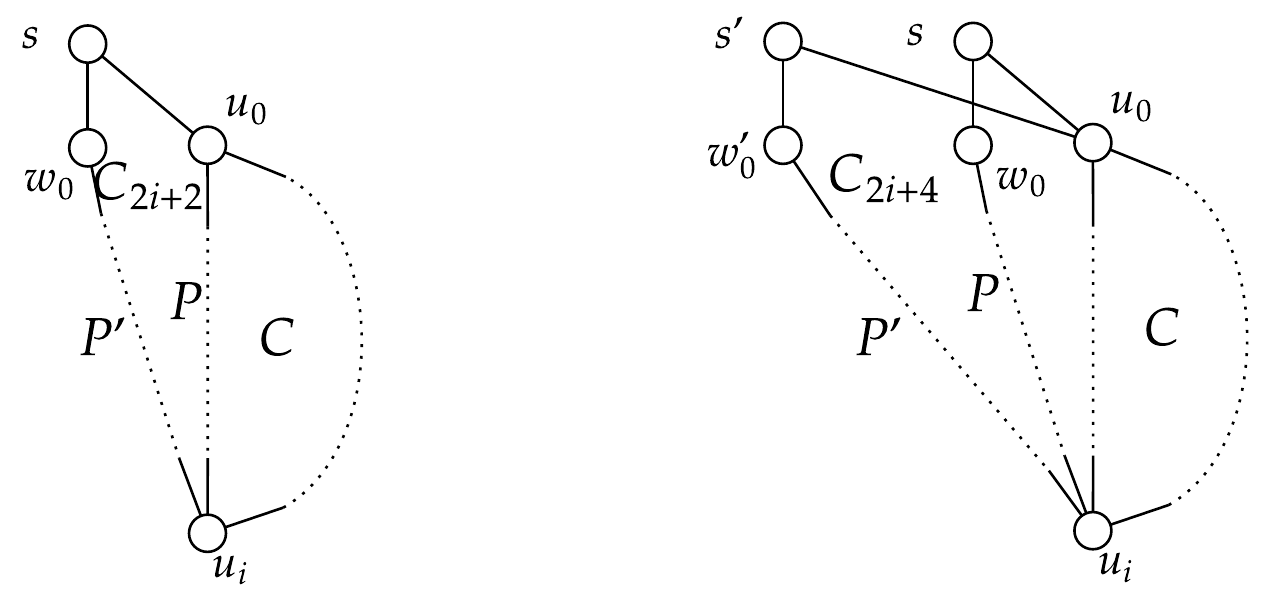}
\caption{Cycles of $\mathcal{F}_{k-1}$ or $4k$-cycles appearing whenever neighbors of $u_0$ have node-disjoint well-colored paths to $u_i$. On the left, $i$ is odd; On the right, $i$ is even.}
\label{fig_incremental}
\end{figure}

\paragraph{Wrap up.}

It results from our analyses for odd and even indices that at most  $k-1$ neighbors $s$ of $u_0$ colored $-1$ may prevent the detection of the $4k$-cycle $(u_0,\dots,u_{4k-1})$, namely at most one source~$s$ for each node $u_{2i}$, for $1\leq i\leq k-1$. Since there are $\Omega(n^{1/2k})$ neighbors of $u_0$ colored~$-1$, the algorithm will  randomly choose at least $k$ different neighbors of~$u_0$ whenever $\Theta(n^{1-1/2k})$ random choices are performed. This completes the proof for $\{C_{4\ell}\mid 1\leq \ell\leq k\}$-freeness, $k\geq 1$. 

\paragraph{$\{C_{4\ell+2}\mid 1\leq \ell \leq k\}$-freeness.}

The proof for $\{C_{4\ell+2}\mid 2\leq \ell \leq k\}$-freeness can be adapted with very little changes from the proof for $\{C_{4\ell}\mid 2\leq \ell\leq k\}$-freeness by inverting the roles of odd and even values of $i$ to compute $T_{2k+1}(i)$.

Indeed, let $i$ be an even index. Similarly to the case of odd indices for $\{C_{4\ell}\mid 2\leq \ell \leq k\}$-freeness, if there are two node-disjoint well-colored paths from a picked $s$ to $u_i$, then $s$ is in a $(2i+2)$-cycle (see Fig.~\ref{fig_incremental}(left)). With $i$ being even, the cycle belongs to $\{C_{4\ell +2}\mid 2\leq \ell \leq k\}$.

On the other hand, let $i$ be an odd index. Similarly to the case of even indices for $\{C_{4\ell}\mid 2\leq \ell\leq k\}$-freeness, if $u_i$ receives more than $(i+1)T_{2k}(i-1)$ identifiers for two different source nodes $s$ and $s'$, then $s$ is in a $(2i+4)$-cycle (see Fig.~\ref{fig_incremental}(right)). With $i$ being odd, the cycle belongs to $\{C_{4l+2}|2\leq l\leq k\}$.

Consequently, the algorithm works by fixing $T_{2k}(0)=1$ and, for $i\geq 1$, 
\[
T_{2k}(i) = \left\{\begin{array}{ll}
        (i+1) \cdot T_{2k}(i-1) & \mbox{if $i$ is odd,}\\
         T_{2k}(i-1) & \mbox{if $i$ is even.}
    \end{array}\right.
\]
This completes the proof of Theorem~\ref{theo:one-every-four}. \qed

\end{document}